\documentclass[journal]{IEEEtran}
\ifCLASSINFOpdf

\else

\fi
\usepackage{amsmath}
\usepackage{makeidx}  
\usepackage{algorithm}
\usepackage{algorithmic}
\usepackage{graphicx}
\usepackage{subfigure}
\usepackage{epstopdf}
\usepackage{bm}
\usepackage{cite}
\usepackage{stfloats}

\usepackage{amssymb}
\usepackage{graphicx}

\usepackage{url}

\usepackage{tabularx,booktabs}
\newcolumntype{C}{>{\centering\arraybackslash}X} 
\usepackage{lipsum}

\usepackage{makecell} 

\usepackage{graphicx}
\usepackage{color}
\newtheorem{thm}{Theorem}
\newtheorem{rem}{Remark}

\newtheorem{pos}{Proposition}
\newtheorem{proof}{proof}

\begin{document}

\title{Multi-IRS Aided ISAC System: Multi-Path Exploitation Versus Reduction}

\author{Guangji Chen,
        Qingqing Wu,
        Shihang Lu,
        Meng Hua,
        and Wen Chen \vspace{-22pt}
        \thanks{Guangji Chen is with Nanjing University of Science and Technology, Nanjing 210094, China (email: guangjichen@njust.edu.cn). Qingqing Wu and Wen Chen are with Shanghai Jiao Tong University, 200240, China (e-mail: qingqingwu@sjtu.edu.cn; wenchen@sjtu.edu.cn). Shihang Lu is with Southern University of Science and Technology, Shenzhen 518055, China(emails: lush2021@mail.sustech.edu.cn). Meng Hua is with Imperial College London, London SW7 2AZ, UK (e-mail: m.hua@imperial.ac.uk).}}

\maketitle
\vspace{-3pt}
\begin{abstract}
This paper investigates a multi-intelligent reflecting surface (IRS) aided integrated sensing and communication (ISAC) system, where multiple IRSs are strategically deployed not only to assist the communication from a multi-antenna base station (BS) to a multi-antenna communication user (CU), but also enable the sensing service for a point target in the non-line-of-sight (NLoS) region of the BS. First, we propose a hybrid multi-IRS architecture, which consists of several passive IRSs and one semi-passive IRS equipped with both active sensors and reflecting elements. To be specific, the active sensors are exploited to receive the echo signals for estimating the target's angle information, and the multiple reflecting paths provided by multi-IRS are employed to improve the degree of freedoms (DoFs) of communication. Under the given budget on the number of total IRSs elements, we theoretically show that increasing the number of deployed IRSs is beneficial for improving DoFs of spatial multiplexing for communication while increasing the  Cr\'amer-Rao bound (CRB) of target estimation, which unveils a fundamental tradeoff between the sensing and communication performance. To characterize the rate-CRB tradeoff, we study a rate maximization problem, by optimizing the BS transmit covariance matrix, IRSs phase-shifts, and the number of deployed IRSs, subject to a maximum CRB constraint. Analytical results reveal that the communication-oriented design becomes optimal when the total number of IRSs elements exceeds
a certain threshold, wherein the relationships of the rate and CRB with the number of IRS elements/sensors, transmit power, and the number of deployed IRSs are theoretically derived and demystified. Simulation results validate our theoretical findings and also demonstrate the superiority of our proposed designs over the benchmark schemes.

\end{abstract}


\IEEEpeerreviewmaketitle

\section{Introduction}
\vspace{-2pt}
Future sixth-generation (6G) wireless networks are expected to support numerous intelligent applications such as autonomous driving, low-altitude economy, and industrial automation, which pose the demand for ultra-high-capacity, ultra-high reliable, and ultra-low-latency communications, as well as the high-resolution/accuracy sensing \cite{zhang2021enabling}. To this end, integrated sensing and communication (ISAC) has been identified as a crucial usage scenario for 6G, which enables the dual use of network infrastructures and wireless signals for providing both the communication and sensing services simultaneously \cite{liu2022integrated}. By coordinating the system design in an integrated hardware platform, ISAC is able to improve the spectrum utilization efficiency and reduce the system cost, thereby promoting both the sensing and communication performance through wireless technologies \cite{liu2022survey}. With advanced technologies such as large antenna arrays and millimeter (mmWave), ISAC has recently drawn tremendous research interest from different research perspectives, e.g., unified ISAC waveforms \cite{liu2022integrated}, wideband ISAC \cite{liu2022integrated}, network-level ISAC \cite{meng2024network}, and transceiver designs for multi-antenna ISAC \cite{lu2024random,liu2020joint}.

Inspired by the spatial multiplexing/diversity gains provided by multiple-input multiple-output (MIMO) techniques, there have been substantial prior works investigating the transceiver designs in MIMO ISAC systems \cite{liu2020joint,lu2022degrees,hua2023optimal,wang2022noma,mu2022noma,liu2021cramer,hua2023mimo,ren2023fundamental}. Due to the scarce power and spectrum resources shared between the functions of sensing and communication, the resource allocations and transmit strategies need to be carefully designed to balance the performance of sensing and communication. However, the transceiver design for MIMO ISAC is challenging owing to the distinct performance metrics for sensing and communications, which leads to the fact that the design principle for the dedicated communication or sensing may not be applicable to ISAC systems \cite{liu2022survey}. Unlike the wide-adopted data rate for the communication metric, the sensing performance metrics highly depend on particular tasks (e.g., target estimation or detection) \cite{liu2022survey}. As a remedy, there are two lines of research works focusing on the sensing metrics of target illumination power for detection \cite{liu2020joint,lu2022degrees,hua2023optimal,wang2022noma,mu2022noma} and Cr\'amer-Rao bound (CRB) for estimation \cite{liu2021cramer,hua2023mimo,ren2023fundamental}, respectively. For instance, the work \cite{hua2023optimal} investigated a MIMO ISAC over broadcast channels by optimizing the transmit beamforming to balance the tradeoff between the target illumination power for sensing and the received signal-to-interference-plus-noise ratio (SINR) for communication. Regarding the estimation task, the literature \cite{liu2021cramer,hua2023mimo,ren2023fundamental} studied the tradeoff between the CRB for target estimation and achievable rate for communication users under the setups of point-to-point MIMO channels, broadcast channels, and multicast channels, respectively.

Despite the research progress on transceiver designs, the performance of ISAC systems is highly affected by the wireless propagation environment. For example, wireless sensing tasks rely on the existence of line-of-sight (LoS) links for extracting useful information on targets, and the communication services require both the high received power and spatial multiplexing degree of freedoms (DoFs) for improving the achievable rate \cite{liu2022integrated}. However, due to randomly distributed obstacles such as trees/buildings, the LoS links from the BS to the communication user (CU)/target may be potentially blocked, thereby degrading both the sensing and communication performance. Hence, relying purely on transceiver design may not provide ubiquitous sensing and communication services over complex channel environments \cite{song2025overview}. To address this challenge faced in ISAC systems, intelligent reflecting surface (IRS) has emerged as a transformative technology to proactively reconfigure wireless channel environments by properly tuning the phase shifts of incident signals on each reflecting element \cite{wu2019beamforming,pan2022overview,di2020smart,chen2022irs1,liu2021reconfigurable,chen2023MA,wu2025intelligent}. In terms of wireless communications, IRS is able to improve the received power \cite{wu2019beamforming,chen2022active,peng2023hybrid,chen2024intelligent1}, increase the channel rank \cite{wu2025intelligent,chen2023channel}, extend the wireless coverage \cite{9351782,chen2023static}, and mitigate co-channel interference through phase-shifts and deployment design \cite{chen2023fundamental,jiang2022interference,chen2024intelligent}. Regarding the wireless sensing, the proper deployment of IRS can create newly virtual LoS links for the targets in the blocked region, which is helpful for addressing the non-line-of-sight (NLoS) sensing issue. Additionally, the high passive beamforming gain provided by the designs of IRS's phase-shifts substantially enhances the echo power, thereby improving the sensing accuracy \cite{song2025overview}.

Given the appealing benefits of IRSs for both the wireless sensing and communications, prior works have studied IRS aided ISAC systems under various setups. Generally, there are two typical IRS architectures for enhancing ISAC, namely fully-passive IRS (P-IRS) and semi-passive IRS (S-IRS) \cite{shao2024intelligent}. The P-IRS operates without dedicated sensors and thus the BS exploits the echo signals that pass through the BS-IRS-target-IRS-BS link for sensing \cite{hua20243d,liu2022joint,zuo2023exploiting,meng2022intelligent,song2024cramer}. In contrast, the S-IRS is equipped with sensors, thereby receiving and processing the echo signals through the BS-IRS-target-IRS link \cite{shao2022target,peng2025semi,wang2023stars,fang2024multi}. By capturing the considerations of the path-loss and passive beamforming gain, the work \cite{song2023fully} provided a theoretical performance comparison for the two architectures and unveiled their respective operating region. From the viewpoint of performance optimization, the works \cite{liu2022joint,zuo2023exploiting,meng2022intelligent,song2024cramer} jointly optimized the BS transmit beamforming and IRS phase-shifts in P-IRS aided ISAC systems by considering different sensing metrics of the target illumination power and CRB. For the S-IRS architecture, prior works \cite{wang2023stars,fang2024multi} considered the employment of the intelligent reflecting/refracting surface and multi-IRS, respectively, to enhance the wireless coverage of ISAC systems.

Despite the aforementioned advancements in IRS aided ISAC, some fundamental issues still remain unaddressed. First, the IRSs deployment guideline for the performance tradeoff between the sensing and communication is unclear. Relative to the IRS phase-shifts design, IRS deployment provides a complementary role in reconfiguring wireless channels \cite{wu2025intelligent}. Due to the different performance metrics of wireless sensing and communication, the two system functionalities may prefer distinct channel conditions. Given the budget on the total number of IRS elements, deploying more small-size IRSs is able to proactively increase the number of paths, thereby improving the spatial multiplexing for MIMO communications. Different from the fact that multiple paths can be exploited to enhance the MIMO communications, all the paths incurred by deploying multi-IRS may be not useful for improving the sensing quality since the sensing task is generally dependent on the LoS path containing the useful information on the target \cite{lu2022degrees}. Therefore, deploying multiple IRSs is expected to have distinct impacts on the sensing and communication performance from the viewpoint of multi-path exploitation versus reduction. Second, the impact of deploying IRSs on the ISAC transceiver designs remains unknown from theoretical perspectives. In ISAC systems without IRS, using communication-oriented transmit signals may not meet the high sensing performance requirement. By deploying massive IRS elements, the resulting favorable propagation environment is expected to enable a high sensing quality even under the communication-oriented design. This is an essential consideration for incorporating the communication-centric designs in ISAC systems.

To shed light on the above considerations, this paper focused on a fundamental multi-IRS aided MIMO ISAC setup with one multi-antenna BS, multiple distributed IRSs, one multi-antenna CU, and one point target in the NLoS region of the BS. In particular, we consider a typical sensing task with prior target knowledge, which corresponds to the target tracking stage for estimating the angle of the point target with respect to an anchor. To facilitate the sensing, we propose a hybrid multi-IRS architecture, which consists of several P-IRSs and one S-IRS equipped with both the active sensors and reflecting elements. In the proposed hybrid multi-IRS architecture, the rich multi-paths environment created by multi-IRS is employed to improve spatial multiplexing and the S-IRS serves as an anchor to process target echo signals for sensing. Additionally, the passive beamforming gain provided by IRSs phase-shifts designs is able to enhance both the received power at the CU and echo signals at the sensors of the S-IRS. The main contributions of this paper are summarized as follows.

\begin{figure}[!t]
\centering
\includegraphics[width= 0.4\textwidth]{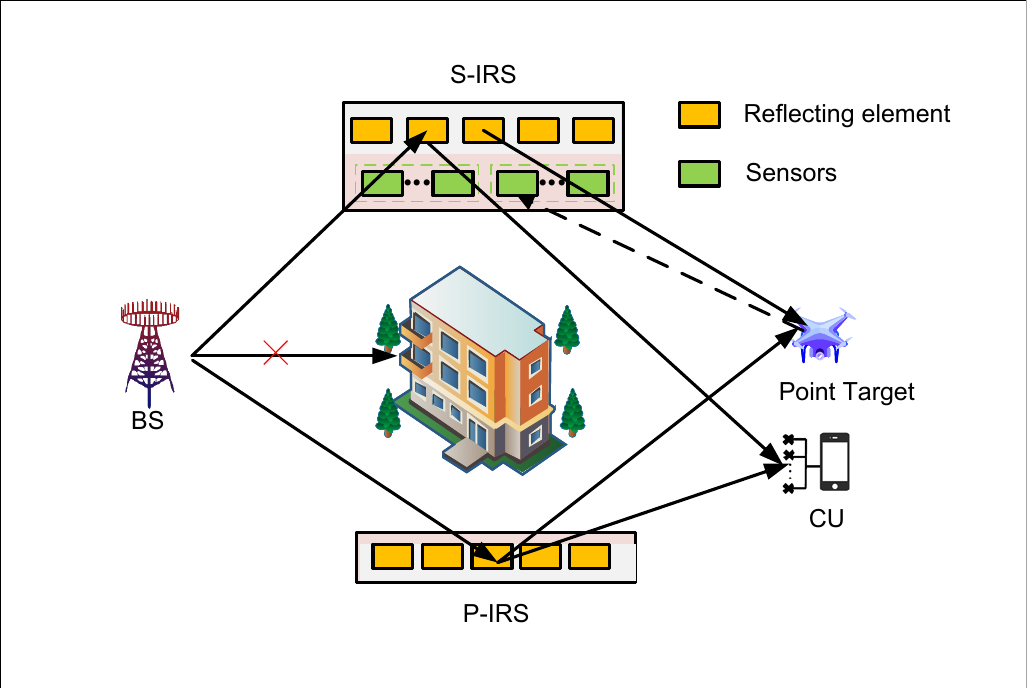}
\DeclareGraphicsExtensions.
\caption{A hybrid multi-IRS aided ISAC system.}
\label{model}
\vspace{-16pt}
\end{figure}
\begin{itemize}
  \item To shed light on the impact of the number of deployed IRSs on the sensing and communication performance, we first derive the closed-form expression of the CRB for target estimation. Based on this, we rigorously prove that the CRB monotonically increases with the number of deployed IRSs while improving DoFs of spatial-multiplexing for communication, which unveils a fundamental tradeoff between the sensing and communication from the viewpoint of multi-path exploitation and reduction.
  \item  To capture the rate-CRB tradeoff, we formulate a rate maximization problem, by optimizing the BS transmit covariance matrix, IRSs phase-shifts, and the number of deployed IRSs, subject to a maximum CRB constraint. By exploiting its particular structure, the problem can be solved in a two-layer manner, where the IRSs phase-shifts and BS transmit covariance matrix are optimized in the inner-layer under the given number of deployed IRSs. In the outer layer, the number of deployed IRSs can be determined via the search over a candidate space.
  \item For the inner optimization problem, we first focus on a special case where the sensing target and CU are co-located to draw useful insights, wherein the optimized BS transmit covariance matrix and IRSs phase-shifts are derived in closed form. Based on this, we theoretically unveil that the communication-oriented design becomes optimal when the total number of IRS elements exceeds a certain threshold, which sheds light on the impact of IRSs on the transceiver design. Then, the relationships of the rate and CRB with the number of IRS elements/sensors, transmit power, and the number of deployed IRSs are theoretically derived.
  \item Finally, an alternating-based successive convex approximation (SCA) algorithm is proposed to extend the results to the general case that the sensing target and CU are separated. Numerical results verify the analytical results and effectiveness of the proposed designs.
\end{itemize}

The remainder of this paper is organized as follows. Section II introduces the system model of the proposed hybrid multi-IRS aided ISAC system. Section III presents the theoretical analysis regarding the impact of the number of deployed IRSs on both the sensing and communication performance, as well as the problem formulation. Section IV provides the solution to the formulated problem and also sheds light on the impact of IRSs on the transceiver designs. Section V presents simulation results to verify the theoretical findings and also to draw useful insights. Finally, Section VI concludes the paper.

\emph{Notations:} Boldface upper-case and lower-case  letter denote matrix and   vector, respectively.  ${\mathbb C}^ {d_1\times d_2}$ stands for the set of  complex $d_1\times d_2$  matrices. For a complex-valued vector $\bf x$, ${\left\| {\bf x} \right\|}$ represents the  Euclidean norm of $\bf x$, ${\rm arg}({\bf x})$ denotes  the phase of   $\bf x$, and ${\rm diag}(\bf x) $ denotes a diagonal matrix whose main diagonal elements are extracted from vector $\bf x$.
For a vector $\bf x$, ${\bf x}^*$ and  ${\bf x}^H$  stand for  its conjugate and  conjugate transpose respectively.   For a square matrix $\bf X$,  ${\rm{Tr}}\left( {\bf{X}} \right)$, $\left\| {\bf{X}} \right\|_2$ and ${\rm{rank}}\left( {\bf{X}} \right)$ respectively  stand for  its trace, Euclidean norm and rank,  while ${\bf{X}} \succeq {\bf{0}}$ indicates that matrix $\bf X$ is positive semi-definite. A circularly symmetric complex Gaussian random variable $x$ with mean $ \mu$ and variance  $ \sigma^2$ is denoted by ${x} \sim {\cal CN}\left( {{{\mu }},{{\sigma^2 }}} \right)$.  ${\cal O}\left(  \cdot  \right)$ is the big-O computational complexity notation.

\section{System model}
As shown in Fig. \ref{model}, we investigate a multi-IRS aided ISAC system, which consists of one millimeter wave (mmWave) BS with ${M_t}$ antennas, multiple IRSs, one CU with ${M_r}$ antennas, and one point target at the NLoS region of the BS. It is assumed that the direct link between the BS and CU/target is blocked due to environmental obstacles, e.g., high buildings. To extend the signal coverage of the BS, multiple IRSs are deployed to assist the data transmission for the CU and the sensing service for the target simultaneously. In the proposed multi-IRS architecture, there are $K$ IRSs, including one S-IRS and $K-1$ P-IRSs. In particular, the S-IRS consists of two parts: reflecting elements and active sensors, where reflecting elements are passive and used for tuning the phases of the incident signal, while the sensors are active and used to receive the echoed signals reflected by the target. All P-IRSs consist of only passive reflecting elements and work in the reflecting mode. Without loss of generality, the $1$-th IRS is referred to the S-IRS, and the $k$-th IRS is referred to the P-IRS with $k \in \left\{ {2, \ldots ,K} \right\}$. Let ${\cal K} \buildrel \Delta \over = \left\{ {1, \ldots ,K} \right\}$ denote the set of IRSs. The S-IRS is equipped with ${N_1}$ reflecting elements and ${N_r}$ sensors, whereas the $k$-th P-IRS, $k \in {\cal K}/\left\{ 1 \right\}$, is equipped with ${N_k}$ reflecting elements. The number of passive elements equipped at the $k$-th IRS is ${N_k} = N/K$, $\forall k$, which satisfies $\sum\nolimits_{k = 1}^K {{N_k}}  = N$ with $N$ denoting the total number of available passive elements.

In particular, we consider an ISAC transmission period consisting of $T$ symbols, and the set of symbols is denoted by ${\cal T} \buildrel \Delta \over = \left\{ {1, \ldots ,T} \right\}$. For achieving full DoFs for simultaneous sensing and communication, the communication signals and dedicated sensing signals are transmitted by the BS. To this end, the signal transmitted by the BS at symbol $t$ is given by
\begin{align}\label{ISAC_signal}
{\bf{x}}\left( t \right) = {{\bf{W}}_c}{{\bf{s}}_c}\left( t \right) + {{\bf{W}}_s}{{\bf{s}}_{{\rm{sen}}}}\left( t \right), t \in {\cal T},
\end{align}
where ${{\bf{s}}_c}\left( t \right) \in {\mathbb{C}^{d \times 1}}$ with $d = \min \left\{ {{M_t},{M_r}} \right\}$ denotes the communication signals intended for the CU satisfying ${{\bf{s}}_c}\left( t \right) \sim {\cal C}{\cal N}\left( {{\bf{0}},{{\bf{I}}_d}} \right)$ and ${{\bf{W}}_c}\in {\mathbb{C}^{{M_t} \times d}}$ represents the associated communication beamformer. Similarly, ${{\bf{s}}_{{\rm{sen}}}}\left( t \right)\in {\mathbb{C}^{{M_t} \times 1}}$ denotes the sensing signal satisfying ${\mathop{\mathbb{E}}\nolimits} \left\{ {{{\bf{s}}_{{\rm{sen}}}}\left( t \right)} \right\} = {{\bf{0}}_{{M_t} \times 1}}$ and ${\mathop{\mathbb{E}}\nolimits} \left\{ {{{\bf{s}}_{{\rm{sen}}}}\left( t \right){\bf{s}}_{{\rm{sen}}}^H\left( t \right)} \right\} = {{\bf{I}}_{{M_t}}}$. Additionally, we assume that communication signals and sensing signals are statistically independent, i.e., ${\mathop{\mathbb{E}}\nolimits} \left\{ {{{\bf{s}}_c}\left( t \right){\bf{s}}_{{\rm{sen}}}^H\left( t \right)} \right\} = {{\bf{0}}_{d \times {M_t}}}$.

Assuming that $T$ is large sufficiently, the sample-wise covariance matrix of the signal ${\bf{x}}\left( t \right)$ can be equivalently replaced by its statistical covariance matrix as
\begin{align}\label{covariance_signal}
{{\bf{R}}_x} &= \frac{1}{T}\sum\limits_{t = 1}^T {\bf{x}} \left( t \right){{\bf{x}}^H}\left( t \right) \approx {\mathbb{E}}\left\{ {{\bf{x}}\left( t \right){{\bf{x}}^H}\left( t \right)} \right\}\nonumber\\
&={{\bf{W}}_c}{\bf{W}}_c^H + {{\bf{W}}_s}{\bf{W}}_s^H = {{\bf{R}}_c} + {{\bf{R}}_s},
\end{align}
where ${{\bf{R}}_c} = {{\bf{W}}_c}{\bf{W}}_c^H$ and ${{\bf{R}}_s} = {{\bf{W}}_s}{\bf{W}}_s^H$. The maximum transmit power of the BS is denoted by ${P_{\max }}$ and thus we have
\begin{align}\label{maximum_power_BS}
{\mathbb{E}}\left\{ {{\bf{x}}\left( t \right){{\bf{x}}^H}\left( t \right)} \right\} = {\mathop{\rm Tr}\nolimits} \left( {{{\bf{R}}_c}} \right) + {\mathop{\rm Tr}\nolimits} \left( {{{\bf{R}}_s}} \right) \le {P_{\max }}.
\end{align}
\subsection{Communication Model}
First, we introduce the communication model from the BS to the CU. Let ${{\bf{G}}_k}\in {\mathbb{C}^{{N_k} \times {M_t}}}$ and ${{\bf{H}}_{{\rm{r}},k}}\in {\mathbb{C}^{{M_r} \times {N_k}}}$ denote the channel matrixes from the BS to the $k$-th IRS and from the $k$-th IRS to the CU, respectively. Considering the fact that the channel power gain of the LoS path is more dominant than that of NLoS paths in mmWave systems, the LoS channel model is adopted in this work. Hence, ${{\bf{G}}_k}$ and ${{\bf{H}}_{{\rm{r}},k}}$ can be expressed as follows
\begin{align}\label{B2I_channel}
{{\bf{G}}_k} = {\rho _{{\rm{BI}},k}}{{\bf{b}}_{{\rm{I}},k}}\left( {\mu _{{\rm{BI}},k}^{\rm{A}}} \right){\bf{a}}_{\rm{B}}^H\left( {\mu _{{\rm{BI}},k}^{\rm{D}}} \right),\forall k,
\end{align}
\begin{align}\label{I2U_channel}
{{\bf{H}}_{{\rm{r}},k}} = {\rho _{{\rm{IU}},k}}{{\bf{a}}_{\rm{U}}}\left( {\mu _{{\rm{IU}},k}^{\rm{A}}} \right){\bf{b}}_{{\rm{I}},k}^H\left( {\mu _{{\rm{IU}},k}^{\rm{D}}} \right),\forall k,
\end{align}
where $\rho _{{\rm{BI}},k}^2$ and $\rho _{{\rm{IU}},k}^2$ denote the large-scale path-loss of the channels from the BS to the $k$-th IRS and from the $k$-th IRS to the CU, respectively. ${\mu _{{\rm{BI}},k}^{\rm{D}}}$ and ${\mu _{{\rm{BI}},k}^{\rm{A}}}$ denote the frequency angle-of-departure (AoD) from the BS to the $k$-th IRS and the angle-of-arrival (AoA) at the IRS, respectively. Similarly, ${\mu _{{\rm{IU}},k}^{\rm{D}}}$ and ${\mu _{{\rm{IU}},k}^{\rm{A}}}$ denote the  frequency angle-of-departure (AoD) from the $k$-th IRS to the CU and the  frequency angle-of-arrival (AoA) at the CU, respectively. In \eqref{B2I_channel} and \eqref{I2U_channel}, ${{\bf{a}}_{\rm{B}}}\left( x \right)$, ${{\bf{b}}_{{\rm{I,}}k}}\left( x \right)$, and ${{\bf{a}}_{\rm{U}}}\left( x \right)$ denote the array response vectors of the BS, the $k$-th IRS, and the CU, respectively, which are given by ${{\bf{a}}_{\rm{B}}}\left( x \right) = {\left[ {{e^{ - j\frac{{{M_t} - 1}}{2}\pi x}}, \ldots ,{e^{j\frac{{{M_t} - 1}}{2}\pi x}}} \right]^T}$, ${{\bf{b}}_{{\rm{I,}}k}}\left( x \right) \!\!=\!\! {\left[ {{e^{ - j\frac{{{N_k} - 1}}{2}\pi x}}, \ldots ,{e^{j\frac{{{N_k} - 1}}{2}\pi x}}} \right]^T}$, and ${{\bf{a}}_{\rm{U}}}\left( x \right) \!\!=\!\! {\left[ {{e^{ - j\frac{{{M_r} - 1}}{2}\pi x}}, \ldots ,{e^{j\frac{{{M_r} - 1}}{2}\pi x}}} \right]^T}$. Let ${{\bf{\Theta }}_k} = {\mathop{\rm diag}\nolimits} \left( {{e^{j{\theta _{k,1}}}}, \ldots ,{e^{j{\theta _{k,{N_k}}}}}} \right)$ denote the reflection matrix of the $k$-th IRS with ${\theta _{k,n}} \in \left[ {0,2\pi } \right)$, $\forall n \in \left\{ {1, \ldots ,{N_k}} \right\}$. Then, the equivalent channel matrix from the BS to the CU can be written as
\begin{align}\label{BS_CU_channel}
{{\bf{H}}_c} = \sum\nolimits_{k = 1}^K {{{\bf{H}}_{{\rm{r}},k}}{{\bf{\Theta }}_k}{{\bf{G}}_k}}.
\end{align}
The received signal by the CU at symbol $t$ is
\begin{align}\label{received_signal}
{{\bf{y}}_c}\left( t \right) = {{\bf{H}}_c}{{\bf{W}}_c}{s_c}\left( t \right) + {{\bf{H}}_c}{{\bf{W}}_s}{s_{\rm{sen}}}\left( t \right) + {{\bf{n}}_c}\left( t \right), t \in {\cal T}
\end{align}
where ${{\bf{n}}_c}\left( t \right) \sim {\cal C}{\cal N}\left( {0,{\sigma ^2}{{\bf{I}}_{{M_r}}}} \right)$ denotes the additive white Gaussian noise (AWGN) at the CU. Note that the sensing signal ${{\bf{x}}_{{\rm{sen}}}}\left( t \right) = {{\bf{W}}_s}{s_{{\rm{sen}}}}\left( t \right)$ is deterministic sequences or pseudo-random and thus it can be known at the CU. Assuming that the CU has the capability to cancel the sensing signal interference, the received signal at the CU can be rewritten as
\begin{align}\label{received_signal1}
{{{\bf{\bar y}}}_c}\left( t \right) = {{\bf{H}}_c}{{\bf{W}}_c}{s_c}\left( t \right) + {{\bf{n}}_c}\left( t \right),t \in {\cal T}.
\end{align}
Therefore, the achievable rate of the CU is given by
\begin{align}\label{achievable_rate}
{R_c} = {\log _2}\det \left( {{{\bf{I}}_{{M_r}}} + \frac{{{{\bf{H}}_c}{{\bf{R}}_c}{\bf{H}}_c^H}}{{{\sigma ^2}}}} \right).
\end{align}
\subsection{Sensing Model}
Next, we illustrate the target sensing model. With dedicated sensors equipped at the S-IRS, NLoS target sensing can be directly performed at the sensors of the S-IRS based on the received signals. The received signals at sensors come from three links: BS-IRSs (passive reflecting elements)-target-sensors link, BS-IRSs-sensors link, and BS-sensors link. To characterize these links, the wireless channels from the $k$-th IRS to the target and from the target to sensors are denoted by ${\bf{h}}_{{\rm{IT,}}k}^H\in {\mathbb{C}^{{1} \times {N_k}}}$ and ${{\bf{h}}_{{\rm{TS}}}}\in {\mathbb{C}^{{N_r} \times {1}}}$, which are given by
\begin{align}\label{target_channel}
{\bf{h}}_{{\rm{IT}},k}^H = {\rho _{{\rm{IT,}}k}}{\bf{b}}_{{\rm{I}},k}^H\left( {\mu _{{\rm{IT}},k}^{\rm{D}}} \right),{{\bf{h}}_{{\rm{TS}}}} = {\rho _{{\rm{TS}}}}{{\bf{a}}_s}\left( {{\mu _{\rm{T}}}} \right).
\end{align}
In \eqref{target_channel}, $\rho _{{\rm{IT}},k}^2$ and $\rho _{{\rm{TS}}}^2$ represent the corresponding large scale path-loss, and ${{\bf{a}}_s}\left( x \right) = {\left[ {{e^{ - j\frac{{{N_r} - 1}}{2}\pi x}}, \ldots ,{e^{j\frac{{{N_r} - 1}}{2}\pi x}}} \right]^T}$. ${\mu _{{\rm{IT}},k}^{\rm{D}}}$ and ${{\mu _{\rm{T}}}}$ are the frequency AoD from the $k$-th IRS to the target and frequency AoA from the target to sensors, respectively. The channels from the BS to sensors and from the $k$-th IRS, $k \in {\cal K}/\left\{ 1 \right\}$, to sensors are denoted by ${{\bf{G}}_{{\rm{BS}}}}\in {\mathbb{C}^{{N_r} \times {M_t}}}$ and ${{\bf{H}}_{{\rm{IS,}}k}}\in {\mathbb{C}^{{N_r} \times {N_k}}}$, respectively, which can be modeled similarly and thus these are omitted here.

Therefore, the received signals by the sensors of the S-IRS in symbol $t$, $t \in {\cal T}$ is given by
\begin{align}\label{received_signals_sensors}
{{\bf{y}}_s}\left( t \right) =& \beta \sum\nolimits_{k = 1}^K {{{\bf{h}}_{{\rm{TS}}}}{\bf{h}}_{{\rm{IT,}}k}^H{{\bf{\Theta }}_k}{{\bf{G}}_k}{\bf{x}}\left( t \right)}  + {{\bf{G}}_{{\rm{BS}}}}{\bf{x}}\left( t \right)\nonumber\\
& + \sum\nolimits_{k = 2}^K {{{\bf{H}}_{{\rm{IS,}}k}}{{\bf{\Theta }}_k}{{\bf{G}}_k}{\bf{x}}\left( t \right)}  + {{\bf{n}}_s}\left( t \right),
\end{align}
where $\beta  \sim {\cal C}{\cal N}\left( {0,1} \right)$ denotes the target radar cross section and ${{\bf{n}}_s}\left( t \right) \sim {\cal C}{\cal N}\left( {0,{\sigma _s^2}{{\bf{I}}_{{N_r}}}} \right)$ represents the noise at the sensors of S-IRS, which may contain the clutter from environment. Note that the locations of the BS and IRSs are fixed, the second term and the third term in the right hand of side in \eqref{received_signals_sensors} can be known by the sensor in advance via offline estimation, which can be perfectly canceled. After removing ${{\bf{G}}_{{\rm{BS}}}}{\bf{x}}\left( t \right)$ and $\sum\nolimits_{k = 2}^K {{{\bf{H}}_{{\rm{IS,}}k}}{{\bf{\Theta }}_k}{{\bf{G}}_k}{\bf{x}}\left( t \right)}$, we can transform \eqref{received_signals_sensors} into
\begin{align}\label{received_signals_sensors_modify}
{{{\bf{\bar y}}}_s}\left( t \right) \!=\! \beta {\rho _{{\rm{TS}}}}{{\bf{a}}_s}\left( {{\mu _{\rm{T}}}} \right)\sum\limits_{k = 1}^K {{\bf{h}}_{{\rm{IT}},k}^H{{\bf{\Theta }}_k}{{\bf{G}}_k}{\bf{x}}\left( t \right)} \!+\! {{\bf{n}}_s}\left( t \right).
\end{align}
The objective of the sensing task  is to estimate the angle information of the target relative to sensors, i.e., ${\mu _{\rm{T}}}$, based on the received signal samples over the all symbols, i.e., ${{{\bf{\bar Y}}}_s} = \left[ {{{{\bf{\bar y}}}_s}\left( 1 \right), \ldots ,{{{\bf{\bar y}}}_s}\left( T \right)} \right]$. To this end, ${\mu _{\rm{T}}}$ can be obtained by employing the celebrated MUSIC algorithm \cite{cheney2001thelinear}.
Let ${{\hat \mu }_{\rm{T}}}$ denote the estimated value of ${\mu _{\rm{T}}}$ \footnote{We focus on a typical target tracking scenario in \cite{liu2021cramer}, especially for static or slow-moving targets. In this case, it is reasonable to assume that the prior information of ${\mu _{\rm{T}}}$ is available to align with existing works \cite{liu2021cramer, ren2023fundamental, song2024cramer}.}.
Typically, the MSE between the estimated ${{\hat \mu }_{\rm{T}}}$ and the real ${\mu _{\rm{T}}}$, i.e., ${\mathbb{E}}\left\{ {{{\left| {{{\hat \mu }_{\rm{T}}} - {\mu _{\rm{T}}}} \right|}^2}} \right\}$, is adopted to measure the sensing performance. Note that it is generally intractable to obtain the exact expression of the MSE ${\mathbb{E}}\left\{ {{{\left| {{{\hat \mu }_{\rm{T}}} - {\mu _{\rm{T}}}} \right|}^2}} \right\}$. As a remedy, we consider the CRB as the performance metric of target sensing, which serves as a lower bound of the MSE, which will be derived for further analysis and problem formulation in the next section.

\section{Problem Formulation}
In this section, we first introduce the proposed spatial multiplexing-oriented IRSs deployment strategy in the considered MIMO ISAC system. Under the proposed IRS deployment strategy, we derive the closed-form expression of CRB by exploiting the specific channel structures. Based on these, we analyze the functional relationships of both the CRB for sensing and multiplexing DoF for communications with respect to the number of deployed IRSs $K$, which sheds light on the fundamental tradeoff between the sensing and communication performance. Then, a joint optimization problem is formulated to maximize the achievable rate of CU while satisfying the CRB threshold of the sensing target.

\subsection{Spatial Multiplexing-Oriented IRS Deployment}
To improve the spatial multiplexing DoF of the CU, we discuss the associated IRSs deployment strategy. To this end, we rewrite ${{\bf{H}}_c}$ as follows
\begin{align}\label{BS_CU_channel_modify}
{{\bf{H}}_c} \!\!&=\!\! \sum\limits_{k = 1}^K {{\rho _{{\rm{c}},k}}{{\bf{a}}_{\rm{U}}}\left( {\mu _{{\rm{IU}},k}^{\rm{A}}} \right)} {\bf{b}}_{{\rm{I}},k}^H\left( {\mu _{{\rm{IU}},k}^{\rm{D}}} \right){{\bf{\Theta }}_k}{{\bf{b}}_{{\rm{I}},k}}\left( {\mu _{{\rm{BI}},k}^{\rm{A}}} \right){\bf{a}}_{\rm{B}}^H\left( {\mu _{{\rm{BI}},k}^{\rm{D}}} \right)\nonumber\\
& = \sum\limits_{k = 1}^K {{{{\bf{\tilde a}}}_{\rm{U}}}\left( {\mu _{{\rm{IU}},k}^{\rm{A}}} \right)} \Upsilon \left( {{{\bf{\Theta }}_k}} \right){\bf{\tilde a}}_{\rm{B}}^H\left( {\mu _{{\rm{BI}},k}^{\rm{D}}} \right),
\end{align}
where ${\rho _{{\rm{c}},k}} = {\rho _{{\rm{BI}},k}}{\rho _{{\rm{IU}},k}}$, ${{{\bf{\tilde a}}}_{\rm{B}}}\left( {\mu _{{\rm{BI}},k}^{\rm{D}}} \right) \!\!\!=\!\! {{\bf{a}}_{\rm{B}}}\left( {\mu _{{\rm{BI}},k}^{\rm{D}}} \right)/\sqrt {{M_t}}$, ${{{\bf{\tilde a}}}_{\rm{U}}}\left( {\mu _{{\rm{IU}},k}^{\rm{A}}} \right) = {{\bf{a}}_{\rm{U}}}\left( {\mu _{{\rm{IU}},k}^{\rm{A}}} \right)/\sqrt {{M_r}}$, $\Upsilon \left( {{{\bf{\Theta }}_k}} \right) \!\!=\!\! \sqrt {{M_r}{M_t}} {\rho _{c,k}}{\bf{b}}_{{\rm{I}},k}^H\left( {\mu _{{\rm{IU}},k}^{\rm{D}}} \right){{\bf{\Theta }}_k}{{\bf{b}}_{{\rm{I}},k}}\left( {\mu _{{\rm{BI}},k}^{\rm{A}}} \right)$. Let ${{\bf{U}}_{\rm{c}}} = \left[ {{{{\bf{\tilde a}}}_{\rm{U}}}\left( {\mu _{{\rm{IU}},1}^{\rm{A}}} \right), \ldots ,{{{\bf{\tilde a}}}_{\rm{U}}}\left( {\mu _{{\rm{IU}},K}^{\rm{A}}} \right)} \right]$, ${{\bf{V}}_{\rm{c}}} = \left[ {{{{\bf{\tilde a}}}_{\rm{B}}}\left( {\mu _{{\rm{BI}},1}^{\rm{D}}} \right), \ldots ,{{{\bf{\tilde a}}}_{\rm{B}}}\left( {\mu _{{\rm{BI}},K}^{\rm{D}}} \right)} \right]$, and ${{\bf{\Lambda }}_{\rm{c}}} = {\mathop{\rm diag}\nolimits} \left( {\Upsilon \left( {{{\bf{\Theta }}_1}} \right), \ldots ,\Upsilon \left( {{{\bf{\Theta }}_K}} \right)} \right)$, ${{\bf{H}}_c}$ can be further expressed as ${{\bf{H}}_c} = {{\bf{U}}_{\rm{c}}}{{\bf{\Lambda }}_{\rm{c}}}{\bf{V}}_{\rm{c}}^H$. Note that ${{\bf{U}}_{\rm{c}}}{{\bf{\Lambda }}_{\rm{c}}}{\bf{V}}_{\rm{c}}^H$ becomes the SVD of ${{\bf{H}}_c}$ provided that ${\bf{V}}_{\rm{c}}^H{{\bf{V}}_{\rm{c}}} = {{\bf{I}}_K}$ and ${\bf{U}}_{\rm{c}}^H{{\bf{U}}_{\rm{c}}} = {{\bf{I}}_K}$. Based on this observation, the deployment of multiple IRSs can be optimized to satisfy the orthogonality among columns in ${{\bf{U}}_{\rm{c}}}$ and ${{\bf{V}}_{\rm{c}}}$, which enables multiple interference-free sub-channels. To achieve this goal, the deployment principle can be obtained in the following proposition.

\begin{pos}
For $K \le \min \left\{ {{M_t},{M_r}} \right\}$, ${\bf{V}}_{\rm{c}}^H{{\bf{V}}_{\rm{c}}} = {{\bf{I}}_K}$ and ${\bf{U}}_{\rm{c}}^H{{\bf{U}}_{\rm{c}}} = {{\bf{I}}_K}$ hold if the following condition is satisfied,
\begin{align}\label{deployment_principle}
&\left| {\mu _{{\rm{BI}},k}^{\rm{D}} - \mu _{{\rm{BI}},i}^{\rm{D}}} \right| = \frac{{2m}}{{{M_t}}},m \in \left\{ {1, \ldots ,{M_t}} \right\}\nonumber\\
&\left| {\mu _{{\rm{IU}},k}^{\rm{A}} - \mu _{{\rm{IU}},i}^{\rm{A}}} \right| = \frac{{2m}}{{{M_r}}},m \in \left\{ {1, \ldots ,{M_r}} \right\}, \forall k \ne i.
\end{align}
\end{pos}
\begin{proof}
Please refer to Appendix A.
\end{proof}

Proposition 1 suggests that the deployment of IRSs can be arranged according to the AoD of the BS (AoA of the CU) presented in \eqref{deployment_principle}, which leads to that  ${{\bf{U}}_{\rm{c}}}{{\bf{\Lambda }}_{\rm{c}}}{\bf{V}}_{\rm{c}}^H$ is the SVD form of ${{\bf{H}}_c}$. In this case, the spatial multiplexing of the CU is obtained as
\begin{align}\label{multiplexing_DoF}
{{\mathop{\rm DoF}\nolimits} _c} &= \mathop {\lim }\limits_{{P_{\max }} \to \infty } \frac{{{R_c}}}{{{{\log }_2}{P_{\max }}}}\nonumber\\
&=\mathop {\lim }\limits_{{P_{\max }} \to \infty } \frac{{\sum\limits_{k = 1}^K {{{\log }_2}\left( {1 + \frac{{{P_{\max }}\Upsilon \left( {{{\bf{\Theta }}_k}} \right)}}{K}} \right)} }}{{{{\log }_2}{P_{\max }}}} = K.
\end{align}
It is observed that increasing the number of IRSs is beneficial for improving the spatial multiplexing DoF of the CU. Based on Proposition 1, a candidate set of the placement of IRSs $\tilde {\cal P}_{{\rm{IRS}}}^{\rm{c}}$ satisfying $\left| {\tilde {\cal P}_{{\rm{IRS}}}^{\rm{c}}} \right| = \min \left\{ {{M_t},{M_r}} \right\}$ and \eqref{deployment_principle} can be obtained.

\subsection{CRB Analysis}
To characterize the sensing performance under the proposed multi-IRS architecture, we obtain the CRB of the estimated ${{\mu _{\rm{T}}}}$ in the following theorem.
\begin{thm}
The CRB of the estimated ${{\mu _{\rm{T}}}}$ is given by
\begin{align}\label{CRB_theta}
&{\mathop{\rm CRB}\nolimits} \left( {{\mu _{\rm{T}}}} \right)\nonumber\\
&= \frac{{\sigma _s^2/\left( {2T{{\left| {\tilde \beta } \right|}^2}} \right)}}{{{\mathop{\rm Tr}\nolimits} \left( {{\bf{\dot A}}\left( {{\mu _{\rm{T}}}} \right){{\bf{R}}_x}{{{\bf{\dot A}}}^H}\left( {{\mu _{\rm{T}}}} \right)} \right) \!\!-\!\! \frac{{{{\left| {{\mathop{\rm Tr}\nolimits} \left( {{\bf{A}}\left( {{\mu _{\rm{T}}}} \right){{\bf{R}}_x}{{{\bf{\dot A}}}^H}\left( {{\mu _{\rm{T}}}} \right)} \right)} \right|}^2}}}{{{\mathop{\rm Tr}\nolimits} \left( {{\bf{A}}\left( {{\mu _{\rm{T}}}} \right){{\bf{R}}_x}{{\bf{A}}^H}\left( {{\mu _{\rm{T}}}} \right)} \right)}}}},
\end{align}
where
\begin{align}\label{def3}
{\left| {\tilde \beta } \right|^2} = \max \left\{ {{{\left| \beta  \right|}^2}\rho _{{\rm{TS}}}^2\rho _{{\rm{IT}},1}^2, \ldots {{\left| \beta  \right|}^2}\rho _{{\rm{TS}}}^2\rho _{{\rm{IT}},K}^2} \right\},
\end{align}
\begin{align}\label{def1}
{\bf{A}}\left( {{\mu _{\rm{T}}}} \right) =& {{\bf{a}}_s}\left( {{\mu _{\rm{T}}}} \right){\bf{b}}_{{\rm{I}},1}^H\left( {{\mu _{\rm{T}}}} \right){{\bf{\Theta }}_1}{{\bf{G}}_1}\nonumber\\
&+ {{\bf{a}}_s}\left( {{\mu _{\rm{T}}}} \right)\sum\nolimits_{k = 2}^K {{\bf{b}}_{{\rm{I}},k}^H\left( {\mu _{{\rm{IT}},k}^{\rm{D}}} \right){{\bf{\Theta }}_k}{{\bf{G}}_k}},
\end{align}
\begin{align}\label{def2}
{\bf{\dot A}}\left( {{\mu _{\rm{T}}}} \right) =&{{{\bf{\dot a}}}_s}\left( {{\mu _{\rm{T}}}} \right)\sum\nolimits_{k = 1}^K {{\rho _{{\rm{BI}},k}}\gamma _k^s{\bf{a}}_{\rm{B}}^H\left( {\mu _{{\rm{BI}},k}^{\rm{D}}} \right)}\nonumber\\
&+ {{\bf{a}}_s}\left( {{\mu _{\rm{T}}}} \right){\rho _{{\rm{BI}},1}}\tilde \gamma _1^s{\bf{a}}_{\rm{B}}^H\left( {\mu _{{\rm{BI}},1}^{\rm{D}}} \right),
\end{align}
with $\gamma _k^s = {\bf{b}}_{{\rm{I}},k}^H\left( {\mu _{{\rm{IT}},k}^{\rm{D}}} \right){{\bf{\Theta }}_k}{{\bf{b}}_{{\rm{I}},k}}\left( {\mu _{{\rm{BI}},k}^{\rm{A}}} \right),\forall k \in {\cal K}$, and $\tilde \gamma _1^s = {\bf{\dot b}}_{{\rm{I}},1}^H\left( {{\mu _{\rm{T}}}} \right){{\bf{\Theta }}_1}{{\bf{b}}_{{\rm{I}},1}}\left( {\mu _{{\rm{BI}},1}^{\rm{A}}} \right)$
\end{thm}
\begin{proof}
Please refer to Appendix B.
\end{proof}

From Theorem 1, ${\mathop{\rm CRB}\nolimits} \left( {{\mu _{\rm{T}}}} \right)$ in \eqref{CRB_theta} reveals the relationship between the covariance of transmitted signal and IRS phase-shift, i.e., ${{{\bf{R}}_x}}$ and $\left\{ {{{\bf{\Theta }}_k}} \right\}$. However, the functional relationship between ${\mathop{\rm CRB}\nolimits} \left( {{\mu _{\rm{T}}}} \right)$ and the number of IRSs $K$ is implicit due to complicated expression in \eqref{CRB_theta}. To shed lights on this issue, we simplify the expression of the CRB under the particularly given ${{{\bf{\Theta }}_1}}$. The results are shown in the following proposition.
\begin{pos}
For the special case that the phase-shift of the S-IRS is set as
\begin{align}\label{phase_SIRS_speical}
{\left[ {{\bf{\Theta }}_1^*} \right]_{n,n}} = \frac{{\left[ {{\rm{diag}}\left( {{\bf{b}}_{{\rm{I}},1}^H\left( {{\mu _{\rm{T}}}} \right)} \right){{\bf{b}}_{{\rm{I}},1}}\left( {\mu _{{\rm{BI}},1}^{\rm{A}}} \right)} \right]_n^*}}{{\left| {{{\left[ {{\rm{diag}}\left( {{\bf{b}}_{{\rm{I}},1}^H\left( {{\mu _{\rm{T}}}} \right)} \right){{\bf{b}}_{{\rm{I}},1}}\left( {\mu _{{\rm{BI}},1}^{\rm{A}}} \right)} \right]}_n}} \right|}},
\end{align}
the CRB of the estimated ${{\mu _{\rm{T}}}}$ is
\begin{align}\label{CRB_special}
{\mathop{\rm CRB}\nolimits} \left( {{\mu _{\rm{T}}}} \right) = \frac{{\sigma _s^2}}{{2T{{\left| {\tilde \beta } \right|}^2}\left\| {{{{\bf{\dot a}}}_s}\left( {{\mu _{\rm{T}}}} \right)} \right\|_2^2{\bf{a}}_{\mathop{\rm B}\nolimits} ^H{{\bf{R}}_x}{{\bf{a}}_{\rm{B}}}}},
\end{align}
where ${\bf{a}}_{\mathop{\rm B}\nolimits} ^H = \sum\nolimits_{k = 1}^K {{\rho _{{\rm{BI}},k}}\gamma _k^s{\bf{a}}_{\rm{B}}^H\left( {\mu _{{\rm{BI}},k}^{\rm{D}}} \right)}$.
\end{pos}
\begin{proof}
Substituting \eqref{phase_SIRS_speical} into $\tilde \gamma _1^s$, we have
\begin{align}\label{gamma_s_speicial}
\tilde \gamma _1^s = {\bf{\dot b}}_{{\rm{I}},1}^H\left( {{\mu _{\rm{T}}}} \right){\bf{\Theta }}_1^*{{\bf{b}}_{{\rm{I}},1}}\left( {\mu _{{\rm{BI}},1}^{\rm{A}}} \right) = {{\bf{1}}^T}{{{\bf{\dot b}}}_{{\rm{I}},1}}\left( 0 \right)=0
\end{align}
Then, ${{\rm{Tr}}\left( {{\bf{\dot A}}\left( {{\mu _{\rm{T}}}} \right){{\bf{R}}_x}{{{\bf{\dot A}}}^H}\left( {{\mu _{\rm{T}}}} \right)} \right)}$, ${{\rm{Tr}}\left( {{\bf{A}}\left( {{\mu _{\rm{T}}}} \right){{\bf{R}}_x}{{{\bf{\dot A}}}^H}\left( {{\mu _{\rm{T}}}} \right)} \right)}$ can be calculated as
\begin{align}\label{pro2_temp1}
&{\rm{Tr}}\left( {{\bf{\dot A}}\left( {{\mu _{\rm{T}}}} \right){{\bf{R}}_x}{{{\bf{\dot A}}}^H}\left( {{\mu _{\rm{T}}}} \right)} \right) = \left\| {{{{\bf{\dot a}}}_s}\left( {{\mu _{\rm{T}}}} \right)} \right\|_2^2{\bf{a}}_{\mathop{\rm B}\nolimits} ^H{{\bf{R}}_x}{{\bf{a}}_{\rm{B}}}, \nonumber\\
&{\rm{Tr}}\left( {{\bf{A}}\left( {{\mu _{\rm{T}}}} \right){{\bf{R}}_x}{{{\bf{\dot A}}}^H}\left( {{\mu _{\rm{T}}}} \right)} \right) = {\mathop{\rm Tr}\nolimits} \left( {{{\bf{a}}_s}\left( {{\mu _{\rm{T}}}} \right){\bf{a}}_{\mathop{\rm B}\nolimits} ^H{{\bf{R}}_x}{{\bf{a}}_{\mathop{\rm B}\nolimits} }{\bf{\dot a}}_s^H\left( {{\mu _{\rm{T}}}} \right)} \right)\nonumber\\
&= {\bf{\dot a}}_s^H\left( {{\mu _{\rm{T}}}} \right){{\bf{a}}_s}\left( {{\mu _{\rm{T}}}} \right){\bf{a}}_{\mathop{\rm B}\nolimits} ^H{{\bf{R}}_x}{{\bf{a}}_{\mathop{\rm B}\nolimits} }\mathop  = \limits^{\left( a \right)} 0.
\end{align}
By plugging \eqref{pro2_temp1} into \eqref{CRB_theta}, \eqref{CRB_special} can be obtained.
\end{proof}

Observing from Proposition 2, the optimization of ${{{\bf{R}}_x}}$ to minimize the CRB in \eqref{CRB_special} is equivalent to maximizing ${\bf{a}}_{\mathop{\rm B}\nolimits} ^H{{\bf{R}}_x}{{\bf{a}}_{\mathop{\rm B}\nolimits} }$. To this end, the optimal ${{{\bf{R}}_x}}$ to maximize ${\bf{a}}_{\mathop{\rm B}\nolimits} ^H{{\bf{R}}_x}{{\bf{a}}_{\mathop{\rm B}\nolimits} }$ is obtained as
\begin{align}\label{optimal_sensing_signal}
{\bf{R}}_{x,s}^* = {P_{\max }}\frac{{{{\bf{a}}_{\rm{B}}}{\bf{a}}_{\rm{B}}^H}}{{{{\left\| {{{\bf{a}}_{\rm{B}}}} \right\|}^2}}}.
\end{align}
By substituting \eqref{optimal_sensing_signal} into \eqref{gamma_s_speicial}, we have
\begin{align}\label{CRB_special2}
&{\mathop{\rm CRB}\nolimits} \left( {{\mu _{\rm{T}}}} \right)\nonumber\\
& = \frac{{\sigma _s^2}}{{2T{{\left| {\tilde \beta } \right|}^2}{P_{\max }}\left\| {{{{\bf{\dot a}}}_s}\left( {{\mu _{\rm{T}}}} \right)} \right\|_2^2{{\left\| {\sum\nolimits_{k = 1}^K {{\rho _{{\rm{BI}},k}}\gamma _k^s{\bf{a}}_{\rm{B}}^H\left( {\mu _{{\rm{BI}},k}^{\rm{D}}} \right)} } \right\|}^2}}}\nonumber\\
&  = \frac{{\sigma _s^2}}{{2T{{\left| {\tilde \beta } \right|}^2}{P_{\max }}\left\| {{{{\bf{\dot a}}}_s}\left( {{\mu _{\rm{T}}}} \right)} \right\|_2^2{M_t}\sum\nolimits_{k = 1}^K {\rho _{{\rm{BI}},k}^2{{\left| {\gamma _k^s} \right|}^2}} }}.
\end{align}

Based on \eqref{CRB_special2}, the monotonic relationship between ${\mathop{\rm CRB}\nolimits} \left( {{\mu _{\rm{T}}}} \right)$ and the number of IRSs $K$ is illustrated in the following proposition.

\begin{pos}
Under the number of total IRS elements budget, i.e., $\sum\nolimits_{k = 1}^K {{N_k}}$, the CRB of the estimated of ${{\mu _{\rm{T}}}}$ in \eqref{CRB_special2} is monotonically increasing with respect to $K$.
\end{pos}
\begin{proof}
Let $f\left( K \right) = \sum\nolimits_{k = 1}^K {\rho _{{\rm{BI}},k}^2} {\left| {\gamma _k^s} \right|^2}$. Then, we show that ${\mathop{\rm CRB}\nolimits} \left( {{\mu _{\rm{T}}}} \right)$ monotonically decreases with respect to $K$ by showing that $f\left( K \right)$ is a decreasing function. Without loss of generality, we assume that $\rho _{{\rm{BI}},1}^2 \ge \rho _{{\rm{BI}},2}^2 \ge  \ldots  \ge \rho _{{\rm{BI}},K}^2$. Under the arbitrarily given number of IRS elements for the $k$-th, i.e., $\left\{ {{N_k}} \right\}$, $f\left( K \right)$ can be obtained as $f\left( K \right) = \sum\nolimits_{k = 1}^K {\rho _{{\rm{BI}},k}^2} N_k^2$. Then, we evaluate the value of $f\left( {K - 1} \right)$. By considering the case of deploying $K-1$ IRSs and adding ${N_K}$ IRS elements at the $1$-th IRS, we have
\begin{align}\label{K_versus K-1}
f\left( {K - 1} \right) &=\sum\nolimits_{k = 2}^{K - 1} {\rho _{{\rm{BI}},k}^2} N_k^2 + \rho _{{\rm{BI}},1}^2{\left( {{N_1} + {N_K}} \right)^2}\nonumber\\
& \ge \sum\nolimits_{k = 2}^{K - 1} {\rho _{{\rm{BI}},k}^2} N_k^2 + \rho _{{\rm{BI}},1}^2\left( {N_1^2 + N_K^2} \right)\nonumber\\
& \mathop  \ge \limits^{\left( a \right)} \sum\nolimits_{k = 1}^K {\rho _{{\rm{BI}},k}^2} N_k^2 = f\left( K \right),
\end{align}
where (a) holds due to $\rho _{{\rm{BI}},1}^2 \ge \rho _{{\rm{BI}},K}^2$. Similarly, it can be shown that $f\left( K \right) \le f\left( {K - 1} \right) \le  \ldots  \le f\left( 1 \right)$, which leads to the desirable result that ${\mathop{\rm CRB}\nolimits} \left( {{\mu _{\rm{T}}}} \right)$ in \eqref{CRB_special2} monotonically increases with respect to $K$.
\end{proof}

Proposition 3 explicitly demonstrates that deploying more distributed IRSs would increase the CRB of the estimated ${{\mu _{\rm{T}}}}$ under the number of total passive elements budget, which implies that one centralized IRS is attractive for improving the sensing performance due to its ability to increase the power of echo signal.

\subsection{Problem Formulation}
From the analysis in the previous subsections, we know that the number of deployed IRSs, i.e., $K$, has different impacts on the communication and sensing performance. In particular, increasing $K$ is helpful for improving the communication DoF of spatial multiplexing while it would degrade the sensing performance. Hence, there exists a fundamental tradeoff between the sensing and communication performance in the considered multi-IRS aided MIMO ISAC system. To capture this performance tradeoff, it is essential to determine the proper number of IRSs, i.e., $K$, to be deployed. Under the given $\tilde {\cal P}_{{\rm{IRS}}}^{\rm{c}}$, the number of deployed IRSs can be controlled by selecting a subset of $\tilde {\cal P}_{{\rm{IRS}}}^{\rm{c}}$, denoted by ${{{\cal P}_{{\rm{IRS}}}}}$, which satisfies $\left| {{{\cal P}_{{\rm{IRS}}}}} \right| = K$. To this end, we consider a two-stage optimization problem with respect to ${{{\cal P}_{{\rm{IRS}}}}}$ and $\left\{ {{{\bf{R}}_c},{{\bf{R}}_s},{{\bf{\Theta }}_k}} \right\}$ to maximize the achievable rate of the CU while satisfying the sensing performance requirement. To this end, the corresponding optimization problem is mathematically expressed as
\begin{subequations}\label{C1}
\begin{align}
\label{C1-a}\mathop {\max }\limits_{{{\cal P}_{{\rm{IRS}}}}}  \;\;& \mathop {\max }\limits_{{{\bf{R}}_c},{{\bf{R}}_s},\left\{ {{{\bf{\Theta }}_k}} \right\}} {\log _2}\det \left( {1 + \frac{{{{\bf{H}}_c}{{\bf{R}}_c}{\bf{H}}_c^H}}{{{\sigma ^2}}}} \right)\\
\label{C1-b}{\rm{s.t.}}\;\;\;&{\mathop{\rm CRB}\nolimits} \left( {{\mu _{\rm{T}}}} \right) \le \varepsilon ,\\
\label{C1-c}&{\mathop{\rm Tr}\nolimits} \left( {{{\bf{R}}_c}} \right) + {\mathop{\rm Tr}\nolimits} \left( {{{\bf{R}}_s}} \right) \le {P_{\max }},\\
\label{C1-d}&{{\bf{R}}_c} \succeq {\bf{0}}, {{\bf{R}}_s} \succeq {\bf{0}},\\
\label{C1-e}&\left| {{{\left[ {{{\bf{\Theta }}_k}} \right]}_{n,n}}} \right| = 1,n \in \left\{ {1, \ldots ,{N_k}} \right\},\forall k \in {\cal K},\\
\label{C1-f}&{{\cal P}_{{\rm{IRS}}}} \subseteq \tilde {\cal P}_{{\rm{IRS}}}^{\rm{c}}.
\end{align}
\end{subequations}
For problem \eqref{C1}, the inner rate maximization problem is over the variables $\left\{ {{{\bf{R}}_c},{{\bf{R}}_s},{{\bf{\Theta }}_k}} \right\}$ under the given ${{\cal P}_{{\rm{IRS}}}}$. and the outer optimization problem is the IRS deployment site selection over the candidate set ${{\cal P}_{{\rm{IRS}}}}$. For the associated constraints in problem \eqref{C1}, \eqref{C1-b} and \eqref{C1-c} are the CRB constraint for the sensing target and the transmit power constraint of the BS, respectively, where $\varepsilon $ represents the CRB threshold for the performance requirement of the sensing task. Constraints \eqref{C1-e} and \eqref{C1-f} indicate the unit-modulus constraint on each reflection coefficient of IRSs and ${{\cal P}_{{\rm{IRS}}}}$ is a subset of $\tilde {\cal P}_{{\rm{IRS}}}^{\rm{c}}$, respectively.

Moreover, problem \eqref{C1} is challenging to be solved due to the coupled optimization variables in the objective function, the complicated constraint in \eqref{C1-b}, as well as the selection of the discrete set ${{\cal P}_{{\rm{IRS}}}}$. Since the total number of available sets ${{\cal P}_{{\rm{IRS}}}}$ is a finite value, the maximum rate of problem \eqref{C1} can be obtained by solving problem \eqref{C1} with any one of sets at first and then selecting the maximum rate among all possible sets satisfying \eqref{C1-f}. Under the given set ${{\cal P}_{{\rm{IRS}}}}$, problem \eqref{C1} reduces to the inner-layer optimization problem with respect to $\left\{ {{{\bf{R}}_c},{{\bf{R}}_s},{{\bf{\Theta }}_k}} \right\}$ as
\begin{subequations}\label{C2}
\begin{align}
\label{C2-a}\mathop {\max }\limits_{{{\bf{R}}_c},{{\bf{R}}_s},\left\{ {{{\bf{\Theta }}_k}} \right\}}\;\;& {\log _2}\det \left( {1 + \frac{{{{\bf{H}}_c}{{\bf{R}}_c}{\bf{H}}_c^H}}{{{\sigma ^2}}}} \right)\\
\label{C2-b}{\rm{s.t.}}\;\;\;\;\;\;\;&\eqref{C1-b}, \eqref{C1-c}, \eqref{C1-d},\eqref{C1-e},
\end{align}
\end{subequations}
where the covariance matrix of transmit signals and IRSs phase-shifts are jointly optimized to maximize the achievable rate of the CU while satisfying the CRB constraint of the target. In the following section, we focus on solving problem \eqref{C2}.

\section{Joint Optimization of Transceiver and IRS Phase-Shift}
In this section, we focus on solving the subproblem of the joint transceiver and IRS phase-shifts design, i.e., $\left\{ {{{\bf{R}}_c},{{\bf{R}}_s},{{\bf{\Theta }}_k}} \right\}$, in \eqref{C2}. Note that problem \eqref{C2} is challenging to be solved optimally because of the tightly coupled variables in the objective function and the complicated expression of CRB, i.e, \eqref{CRB_theta}, in constraint \eqref{C1-b}. As suggested by Proposition 2, the CRB \eqref{CRB_special} under the specific S-IRS's phase-shifts monotonically decreases with respect to the value of ${\bf{a}}_{\rm{B}}^H\left( {{{\bf{R}}_c} + {{\bf{R}}_s}} \right){{\bf{a}}_{\rm{B}}}$. Therefore, we first transform the original problem \eqref{C2} into
\begin{subequations}\label{C4}
\begin{align}
\label{C4-a}\mathop {\max }\limits_{{{\bf{R}}_c},{{\bf{R}}_s},\left\{ {{{\bf{\Theta }}_k}} \right\}}\;\;& {\log _2}\det \left( {1 + \frac{{{{\bf{H}}_c}{{\bf{R}}_c}{\bf{H}}_c^H}}{{{\sigma ^2}}}} \right)\\
\label{C4-b}{\rm{s.t.}}\;\;\;\;\;\;\;&{\bf{a}}_{\rm{B}}^H\left( {{{\bf{R}}_c} + {{\bf{R}}_s}} \right){{\bf{a}}_{\rm{B}}} \ge {\Gamma _s},\\
\label{C4-c}&\eqref{C1-c}, \eqref{C1-d},\eqref{C1-e},
\end{align}
\end{subequations}
where ${\Gamma _s} = \sigma _s^2/\left( {2T{{\left| {\tilde \beta } \right|}^2}\left\| {{{{\bf{\dot a}}}_s}\left( {{\mu _{\rm{T}}}} \right)} \right\|_2^2\varepsilon } \right)$.
\begin{rem}
Before proceeding to solve problem \eqref{C4}, we discuss the two special cases with sole communication service and sole target sensing service, respectively. First, considering the communication service only, the maximum achievable rate can be obtained by solving the following problem:
\begin{subequations}\label{C5}
\begin{align}
\label{C5-a}\mathop {\max }\limits_{{{\bf{R}}_c},\left\{ {{{\bf{\Theta }}_k}} \right\}}\;\;& {\log _2}\det \left( {1 + \frac{{{{\bf{H}}_c}{{\bf{R}}_c}{\bf{H}}_c^H}}{{{\sigma ^2}}}} \right)\\
\label{C5-b}{\rm{s.t.}}\;\;\;\;\;&{\mathop{\rm Tr}\nolimits} \left( {{{\bf{R}}_c}} \right) \le {P_{\max }},{{\bf{R}}_c} \succeq {\bf{0}},\\
\label{C5-c}&\eqref{C1-e}.
\end{align}
\end{subequations}
It is not difficult to show that the optimal solution of problem \eqref{C5} is given by
\begin{align}\label{phase_IRS_opt_com}
{\left[ {{\bf{\Theta }}_k^{c*}} \right]_{n,n}} = \frac{{\left[ {{\rm{diag}}\left( {{\bf{b}}_{{\rm{I}},k}^H\left( {\mu _{{\rm{IU}},k}^{\rm{D}}} \right)} \right){{\bf{b}}_{{\rm{I}},k}}\left( {\mu _{{\rm{BI}},k}^{\rm{A}}} \right)} \right]_n^*}}{{\left| {{{\left[ {{\rm{diag}}\left( {{\bf{b}}_{{\rm{I}},k}^H\left( {\mu _{{\rm{IU}},k}^{\rm{D}}} \right)} \right){{\bf{b}}_{{\rm{I}},k}}\left( {\mu _{{\rm{BI}},k}^{\rm{A}}} \right)} \right]}_n}} \right|}},
\end{align}
\begin{align}\label{transmit_opt_com}
{\bf{R}}_c^* = \sum\nolimits_{k = 1}^K {p_{c,k}^*{{{\bf{\tilde a}}}_{\rm{B}}}\left( {\mu _{{\rm{BI}},k}^{\rm{D}}} \right)} {\bf{\tilde a}}_{\rm{B}}^H\left( {\mu _{{\rm{BI}},k}^{\rm{D}}} \right),
\end{align}
where $p_{c,k}^* = {\left( {\nu  - {\sigma ^2}/{\Upsilon ^2}\left( {{\bf{\Theta }}_k^{c*}} \right)} \right)^ + },\forall k$ with $\nu $ denoting the water level which can be obtained according to $\sum\nolimits_{k = 1}^K {p_{c,k}^*}  = {P_{\max }}$. With the obtained $\left\{ {{\bf{\Theta }}_k^{c*},{\bf{R}}_c^*} \right\}$, the maximum achievable rate of the CU is ${R_{\max }} = \sum\nolimits_{k = 1}^K {{{\log }_2}\left( {1 + \frac{{p_{c,k}^*{\Upsilon ^2}\left( {{\bf{\Theta }}_k^{c*}} \right)}}{{{\sigma ^2}}}} \right)} $. Next, we consider the special case of target sensing only and the associated CRB minimization problem is given by
\begin{subequations}\label{C6}
\begin{align}
\label{C6-a}\mathop {\max }\limits_{{{\bf{R}}_s},\left\{ {{{\bf{\Theta }}_k}} \right\}}\;\;&{\bf{a}}_{\rm{B}}^H{{\bf{R}}_s}{{\bf{a}}_{\rm{B}}}\\
\label{C6-b}{\rm{s.t.}}\;\;\;\;\;&{\mathop{\rm Tr}\nolimits} \left( {{{\bf{R}}_s}} \right) \le {P_{\max }},{{\bf{R}}_s} \succeq {\bf{0}},\\
\label{C6-c}&\eqref{C1-e}.
\end{align}
\end{subequations}
The optimal solution of problem \eqref{C6} can be derived as
\begin{align}\label{phase_IRS_opt_sensing}
{\left[ {{\bf{\Theta }}_k^{{\rm{s}}*}} \right]_{n,n}} = \frac{{\left[ {{\rm{diag}}\left( {{\bf{b}}_{{\rm{I}},k}^H\left( {\mu _{{\rm{IT}},k}^{\rm{D}}} \right)} \right){{\bf{b}}_{{\rm{I}},k}}\left( {\mu _{{\rm{BI}},k}^{\rm{A}}} \right)} \right]_n^*}}{{\left| {{{\left[ {{\rm{diag}}\left( {{\bf{b}}_{{\rm{I}},k}^H\left( {\mu _{{\rm{IT}},k}^{\rm{D}}} \right)} \right){{\bf{b}}_{{\rm{I}},k}}\left( {\mu _{{\rm{BI}},k}^{\rm{A}}} \right)} \right]}_n}} \right|}},
\end{align}
\begin{align}\label{transmit_opt_sensing}
{\bf{R}}_s^* = {P_{\max }}\frac{{{{\bf{a}}_{\rm{B}}}{\bf{a}}_{\rm{B}}^H}}{{\left\| {{{\bf{a}}_{\rm{B}}}} \right\|_2^2}}.
\end{align}
With the obtained $\left\{ {{\bf{\Theta }}_k^{{\rm{s}}*},{\bf{R}}_s^*} \right\}$, the resulting CRB can be obtained in \eqref{CRB_special2}.
\end{rem}

Motivated by the discussion in Remark 1, we construct the covariance matrixes of the transmitted communication signals and sensing signals as
\begin{align}\label{transmit_signal_structure}
{{\bf{R}}_c} = \sum\limits_{k = 1}^K {{p_{c,k}}{{{\bf{\tilde a}}}_{\rm{B}}}\left( {\mu _{{\rm{BI}},k}^{\rm{D}}} \right)} {\bf{\tilde a}}_{\rm{B}}^H\left( {\mu _{{\rm{BI}},k}^{\rm{D}}} \right),{{\bf{R}}_s} = {p_s}\frac{{{{\bf{a}}_{\rm{B}}}{\bf{a}}_{\rm{B}}^H}}{{\left\| {{{\bf{a}}_{\rm{B}}}} \right\|_2^2}}.
\end{align}
Based on \eqref{transmit_signal_structure}, ${\bf{a}}_{\rm{B}}^H\left( {{{\bf{R}}_c} + {{\bf{R}}_s}} \right){{\bf{a}}_{\rm{B}}}$ can be expressed as
\begin{align}\label{sensing_power}
&{\bf{a}}_{\rm{B}}^H\left( {{{\bf{R}}_c} + {{\bf{R}}_s}} \right){{\bf{a}}_{\rm{B}}}\nonumber\\
& = {p_s}{M_t}\sum\nolimits_{k = 1}^K {\rho _{{\rm{BI}},k}^2{{\left| {\gamma _k^s} \right|}^2}}  + {M_t}\sum\nolimits_{k = 1}^K {\rho _{{\rm{BI}},k}^2{{\left| {\gamma _k^s} \right|}^2}} {p_{c,k}}\nonumber\\
& \buildrel \Delta \over = {f_s}\left( {\left\{ {{p_{c,k}}} \right\},{p_s}} \right)
\end{align}

Then, problem \eqref{C4} can be transformed into the joint power allocation $\left\{ {{p_{c,k}},{p_s}} \right\}$ and IRSs' phase-shifts ${{\bf{\Theta }}_k}$ optimization problem as
\begin{subequations}\label{C7}
\begin{align}
\label{C7-a}\mathop {\max }\limits_{\left\{ {{p_{c,k}},{p_s}} \right\},\left\{ {{{\bf{\Theta }}_k}} \right\}}\;\;&\sum\nolimits_{k = 1}^K {{{\log }_2}\left( {1 + \frac{{{p_{c,k}}{\Upsilon ^2}\left( {{{\bf{\Theta }}_k}} \right)}}{{{\sigma ^2}}}} \right)} \\
\label{C7-b}{\rm{s.t.}}\;\;\;\;\;\;\;\;\;&{f_s}\left( {\left\{ {{p_{c,k}}} \right\},{p_s}} \right) \ge {\Gamma _s},\\
\label{C7-c}&\sum\nolimits_{k = 1}^K {{p_{c,k}}}  + {p_s} \le {P_{\max }},\\
\label{C7-d}&{p_{c,k}} \ge 0, {p_{s}} \ge 0,\\
\label{C7-e}&\eqref{C1-e}.
\end{align}
\end{subequations}

\subsection{Special Case Study: Co-located CU and Target}
To gain more useful insights, we first study a special case where the served CU is also a target to be sensed. One typical example of this scenario is sensing-assisted communication in the vehicle-to-infrastructure application, where a BS needs to communicate with a CU while tracking its movement simultaneously \cite{lu2022degrees}. In this case, we have $\mu _{{\rm{IU,}}k}^{\rm{D}} = \mu _{{\rm{IT,}}k}^{\rm{D}},\forall k$, which leads to that the optimal IRSs phase-shifts, denoted by $\left\{ {{\bf{\Theta }}_k^*} \right\}$, satisfy the condition ${\bf{\Theta }}_k^* = {\bf{\Theta }}_k^{c*} = {\bf{\Theta }}_k^{s*}$. Recall that the closed-form expressions of ${\bf{\Theta }}_k^{c*}$ and ${\bf{\Theta }}_k^{s*}$ are given in \eqref{phase_IRS_opt_com} and \eqref{phase_IRS_opt_sensing}, respectively. By substituting the optimal ${\bf{\Theta }}_k^*$ into the objective function \eqref{C7-a} and constraint \eqref{C7-b}, problem \eqref{C7} can be equivalently transformed into
\begin{subequations}\label{C8}
\begin{align}
\label{C8-a}\mathop {\max }\limits_{\left\{ {{p_{c,k}}} \right\},{p_s}} \;\;&\sum\limits_{k = 1}^K {{{\log }_2}\left( {1 + \frac{{{M_t}{M_r}{p_{c,k}}\rho _{{\rm{BI}},k}^2\rho _{{\rm{IU}},k}^2N_k^2}}{{{\sigma ^2}}}} \right)}\\
\label{C8-b}{\rm{s.t.}}\;\;\;\;\;&{p_s}{M_t}\sum\nolimits_{k = 1}^K {\rho _{{\rm{BI}},k}^2N_k^2} \nonumber\\
& + {M_t}\sum\nolimits_{k = 1}^K {\rho _{{\rm{BI}},k}^2N_k^2} {p_{c,k}} \ge {\Gamma _s},\\
\label{C8-c}&\eqref{C7-c}, \eqref{C7-d}.
\end{align}
\end{subequations}
In problem \eqref{C8}, the objective function is concave with respect to ${\left\{ {{p_{c,k}}} \right\}}$ and the constraints in \eqref{C7-c}, \eqref{C7-d} are both affine. As a result, problem \eqref{C8} is a convex optimization problem, which can be solved optimally by using the Lagrange duality method. Denote the optimal solution of problem \eqref{C8} as $\left\{ {p_{c,k}^ \star ,p_s^ \star } \right\}$. The closed-form expressions of the optimal solution to problem \eqref{C8} are derived in the following theorem.
\begin{thm}
If the following condition
\begin{align}\label{sensing_signal_condition11}
\sum\limits_{k = 1}^K {{{\left[ {\frac{1}{{\mu^\star {M_t}\sum\nolimits_{l \ne k}^K {\rho _{{\rm{BI}},l}^2N_l^2} }} \!\!-\!\! \frac{{{\sigma ^2}}}{{{M_t}{M_r}\rho _{{\rm{BI}},k}^2\rho _{{\rm{IU}},k}^2N_k^2}}} \right]}^ + }} \!\! <\!\! {P_{\max }}
\end{align}
is satisfied, the optimal power allocation is given by
\begin{align}\label{power_com11}
p_{c,k}^ \star  &= p_{c,k}^{\left( {\rm{I}} \right) }\left( {{\mu^\star}} \right)\nonumber\\
& = {\left[ {\frac{1}{{\mu^\star{M_t}\sum\nolimits_{l \ne k}^K {\rho _{{\rm{BI}},l}^2N_l^2} }} \!\!-\!\! \frac{{{\sigma ^2}}}{{{M_t}{M_r}\rho _{{\rm{BI}},k}^2\rho _{{\rm{IU}},k}^2N_k^2}}} \right]^ + }\nonumber\\
&\hspace{-0.5cm}p_s^ \star  = {P_{\max }} - \sum\nolimits_{k = 1}^K {p_{c,k}^{\left( {\rm{I}} \right)}\left( {{\mu ^ \star }} \right)},
\end{align}
where ${{\mu ^ \star }}$ is the unique root of
\begin{align}\label{equation1}
{{\cal G}_{{\rm{s,1}}}}\left( \mu  \right)  \buildrel \Delta \over = & {M_t}\sum\nolimits_{k = 1}^K {\left( {{P_{\max }} - \sum\nolimits_{l \ne k}^K {p_{c,l}^{\left( {\rm{I}} \right)}\left( \mu  \right)} } \right)\rho _{{\rm{BI}},k}^2N_k^2} \nonumber\\
& - {\Gamma _s} = 0.
\end{align}
In contrast, if condition \eqref{sensing_signal_condition11} is not satisfied, we have
\begin{align}\label{power_com12}
&p_{c,k}^ \star  = p_{c,k}^{{\rm{II}}}\left( {{\nu ^ \star }} \right)={\left[ {\frac{1}{\nu ^ \star} - \frac{{{\sigma ^2}}}{{{M_t}{M_r}\rho _{{\rm{BI}},k}^2\rho _{{\rm{IU}},k}^2N_k^2}}} \right]^ + }\nonumber\\
&p_s^ \star = 0,
\end{align}
provided that
\begin{align}\label{sensing_signal_condition22}
\sum\limits_{k = 1}^K {{{\left[ {\frac{{{M_t}\rho _{{\rm{BI}},k}^2N_k^2}}{{{\nu ^ \star }}} - \frac{{{\sigma ^2}}}{{{M_r}\rho _{{\rm{IU}},k}^2}}} \right]}^ + } \ge {\Gamma _s}},
\end{align}
where $\nu^ \star$ is the unique root of $\sum\nolimits_{k = 1}^K {p_{c,k}^{{\rm{II}}}\left( \nu  \right)} = {P_{\max }}$. Otherwise, we have
\begin{align}\label{power_com13}
&p_{c,k}^ \star  = p_{c,k}^{{\rm{III}}}\left( {{\mu ^ \star },{\nu ^ \star }} \right) \nonumber\\
& = {\left[ {\frac{1}{{{\mu ^ \star }{M_t}\sum\nolimits_{l \ne k}^K {\rho _{{\rm{BI}},l}^2N_l^2}  + {\nu ^ \star }}} - \frac{{{\sigma ^2}}}{{{M_t}{M_r}\rho _{{\rm{BI}},k}^2\rho _{{\rm{IU}},k}^2N_k^2}}} \right]^ + },\nonumber\\
&p_s^ \star = 0,
\end{align}
where $\nu^ \star$ and $\mu^ \star$ are unique roots of  $\sum\nolimits_{k = 1}^K {p_{c,k}^ {{\rm{III}}} \left( {\mu ,\lambda } \right)}  = {P_{\max }}$ and $\sum\nolimits_{k = 1}^K {p_{c,k}^{{\rm{III}}}\left( {\mu ,\lambda } \right)} {M_t}\rho _{{\rm{BI}},k}^2N_k^2 = {\Gamma _s}$
\end{thm}
\begin{proof}
Please refer to Appendix C.
\end{proof}

Theorem 2 unveils the optimal power allocation structure for balancing the tradeoff of communication and sensing. Note that \eqref{sensing_signal_condition11} in Theorem 2 serves as a necessary and sufficient condition for ensuring that a dedicated sensing signal is activated. When \eqref{sensing_signal_condition11} is satisfied, the optimal power allocation for information data streams admits a multi-level water-filling structure, and the remaining power is allocated to the dedicated sensing space for enhancing the sensing power. In this case, the sensing power constraint \eqref{C8-b} is met with equality. Under the condition of \eqref{sensing_signal_condition11} is not satisfied, the dedicated sensing signal is not needed, thereby allocating all the power to the communication space is sufficient for satisfying the performance requirement of the sensing task. In this case, the optimal tradeoff between sensing and communication is achieved by varying the water-levels allocated into different sub-channels provided that \eqref{sensing_signal_condition22} is not satisfied, as illustrated in \eqref{power_com13}. In contrast, when the condition \eqref{sensing_signal_condition22} is satisfied, it is interesting to observe that the optimal power allocation in \eqref{power_com12} reduces to the conventional water-filling structure with a constant power level.

\begin{rem}
The dedicated sensing signal activation condition in \eqref{sensing_signal_condition11} has an interesting interpretation. In particular, given the system parameters (i.e., ${P_{\max }}$, ${\Gamma _s}$, and ${M_t}$) and the channel setup, a virtual transmit power $p_{c,k}^{\left( {\rm{I}} \right) }\left( {{\mu^\star}} \right)$ for each data stream according to \eqref{power_com11}. By letting the BS transmit at this power, i.e.,  $p_{c,k}^{\left( {\rm{I}} \right) }\left( {{\mu^\star}} \right)$, the total transmit power allocated to the communication signal can be obtained as $\sum\nolimits_{k = 1}^K {p_{c,k}^{\left( {\rm{I}} \right)}\left( {{\mu ^ \star }} \right)}$. As a result, ${P_{\max }} > \sum\nolimits_{k = 1}^K {p_{c,k}^{\left( {\rm{I}} \right)}\left( {{\mu ^ \star }} \right)}$ means that using communication-oriented signal only is not sufficient to meet the sensing performance requirement, which leads to the fact that the dedicated sensing signal is needed to further improve the echo power. In contrast, ${P_{\max }} \le \sum\nolimits_{k = 1}^K {p_{c,k}^{\left( {\rm{I}} \right)}\left( {{\mu ^ \star }} \right)}$ means that exploiting purely communication-oriented signal is already sufficient for achieving the sensing requirement and hence dedicated sensing signal is not needed to save transmit power for improving the communication performance. Therefore, whether the dedicated sensing signal is needed or not, fundamentally depends on whether the available transmit power is the bottleneck for meeting the sensing performance requirement. Let
\begin{align}\label{remaining_power}
{\cal F}\left( {{P_{\max }},{M_t},{\Gamma _s}} \right) = {P_{\max }} - \sum\nolimits_{k = 1}^K {p_{c,k}^{\left( {\rm{I}} \right)}\left( {{\mu ^ \star }} \right)}.
\end{align}
\end{rem}
It is evident that ${\cal F}\left( {{P_{\max }},{M_t},{\Gamma _s}} \right)$ increases and decreases with $\left\{ {{P_{\max }},{M_t}} \right\}$ and ${\Gamma _s}$, respectively. The results agree with the intuition that low ${{P_{\max }},{M_t}}$ or high ${{\Gamma _s}}$ will make dedicated sensing signal more likely to be activated and vice versa.

Then, we focus on discussing the impact of the number of total IRS elements on the optimal power allocation structure. Recall that ${N_k} = N/K,\forall k$, and then the following proposition is obtained.
\begin{pos}
The optimal solution of problem \eqref{C8} is given by
\begin{align}\label{power_com_special}
&p_{c,k}^ \star  = \frac{{{P_{\max }}}}{K} - \frac{{{K^2}{\sigma ^2}}}{{{M_t}{M_r}{N^2}}}{\xi _k} > 0, \forall k, ~~p_s^ \star = 0,
\end{align}
if and only if the following condition is satisfied
\begin{align}\label{condition_optimal}
N \!>\! \max \left\{ {\sqrt {\frac{{{K^3}\left( {{\Gamma _s} \!\!+\!\! \frac{{{\sigma ^2}}}{{{M_r}}}\sum\limits_{k = 1}^K {{\xi _k}\rho _{{\rm{BI}},k}^2} } \right)}}{{{P_{\max }}{M_t}\sum\limits_{k = 1}^K {\rho _{{\rm{BI}},k}^2} }}} ,\sqrt {\frac{{{K^3}\xi {\sigma ^2}}}{{{P_{\max }}{M_t}{M_r}}}} } \right\},
\end{align}
where
\begin{align}\label{diff_pathloss}
{\xi _k} = \frac{1}{{\rho _{{\rm{BI}},k}^2\rho _{{\rm{IU}},k}^2}} - \frac{1}{K}\sum\nolimits_{k = 1}^K {\frac{1}{{\rho _{{\rm{BI}},k}^2\rho _{{\rm{IU}},k}^2}}}
\end{align}
and $\xi  = \mathop {\max }\limits_{k \in {\cal K}} \left\{ {{\xi _k}} \right\}$.
\end{pos}
\begin{proof}
First, to ensure that the maximum number of transmitted data streams is $K$, i.e., $p_{c,k}^ \star  > 0, \forall k$, the water-level ${{\nu ^ \star }}$ can be derived as
\begin{align}\label{water-level}
{{\nu ^ \star } = \frac{K}{{{P_{\max }} + \frac{{{K^2}{\sigma ^2}}}{{{M_t}{M_r}{N^2}}}\sum\nolimits_{k = 1}^K {\frac{1}{{\rho _{{\rm{BI}},k}^2\rho _{{\rm{IU}},k}^2}}} }}}.
\end{align}
Based on \eqref{water-level} and $p_{c,k}^ \star  > 0, \forall k$, we have
\begin{align}\label{full_stream}
\frac{{{P_{\max }}}}{K} - \frac{{{K^2}{\sigma ^2}}}{{{M_t}{M_r}{N^2}}}{\xi _k} > 0,\forall k,
\end{align}
which directly leads to
\begin{align}\label{N_threshold1}
{N > \sqrt {\frac{{{K^3}\xi {\sigma ^2}}}{{{P_{\max }}{M_t}{M_r}}}} }.
\end{align}
Then, by exploiting the result in \eqref{sensing_signal_condition22}, we have
\begin{align}\label{sensing_requirement}
\frac{{{P_{\max }}}}{{{K^3}}}{M_t}{N^2}\sum\nolimits_{k = 1}^K {\rho _{{\rm{BI}},k}^2 - } \frac{{{\sigma ^2}}}{{{M_r}}}\sum\nolimits_{k = 1}^K {\rho _{{\rm{BI}},k}^2{\xi _k} \ge {\Gamma _s}}.
\end{align}
Thus, we obtain
\begin{align}\label{N_threshold2}
N \ge \sqrt {\frac{{{K^3}\left( {{\Gamma _s} + \frac{{{\sigma ^2}}}{{{M_r}}}\sum\nolimits_{k = 1}^K {{\xi _k}\rho _{{\rm{BI}},k}^2} } \right)}}{{{P_{\max }}{M_t}\sum\nolimits_{k = 1}^K {\rho _{{\rm{BI}},k}^2} }}}.
\end{align}
By combining \eqref{N_threshold1} and \eqref{N_threshold2}, we obtain the result in \eqref{power_com_special}, which thus completes the proof.
\end{proof}

Proposition 4 directly sheds light on how many IRS elements are needed to ensure that conventional water-filling power allocations are optimal for achieving the maximum communication performance while satisfying the sensing performance requirement. When condition \eqref{condition_optimal} is satisfied, exploiting purely communication-optimal signals defined in \eqref{transmit_opt_com} may not lose the optimality of the original problem \eqref{C2}, which also implies that a dedicated sensing signal is not needed. Benefiting from the high passive beamforming gain introduced by a massive number of IRS elements, using a pure communication-optimal signal is sufficient for meeting the performance requirement of target sensing.

\begin{rem}
To capture the communication and sensing performance in a large number of IRS elements regime, i.e., $N \to \infty $, we have
\begin{align}\label{power_com_large_N}
&p_{c,k}^ \star  = \frac{{{P_{\max }}}}{K}, \forall k, ~~p_s^ \star = 0
\end{align}
according to the result in \eqref{power_com_special}. In this case, the maximum achievable rate of the proposed architecture is derived as
\begin{align}\label{rate_large_N}
{R_c} = \sum\nolimits_{k = 1}^K {{{\log }_2}\left( {1 + \frac{{{P_{\max }}\rho _{{\rm{BI,}}k}^2\rho _{{\rm{BI,}}k}^2{M_t}{M_r}}}{{{K^3}}}{N^2}} \right)}.
\end{align}
Thus, the scaling law of $R$ with respect to $N$ and $K$ can be obtained as
\begin{align}\label{rate_scaling_law}
{R_c} &\cong {\log _2}\left( {\prod\limits_{k = 1}^K {{P_{\max }}\rho _{{\rm{BI,}}k}^2\rho _{{\rm{BI,}}k}^2{M_t}{M_r}} } \right) + K{\log _2}\frac{{{N^2}}}{{{K^3}}}\nonumber\\
& \sim {\cal{O}}\left( {K{{\log }_2}\left( {{N^2}/{K^3}} \right)} \right).
\end{align}
Regarding the sensing performance, the closed-form expression of the CRB can be expressed as
\begin{align}\label{CRB_large_N}
{\rm{CRB}}\left( {{\mu _{\rm{T}}}} \right) &= \frac{{6\sigma _s^2{T^{ - 1}}{{\left| {\tilde \beta } \right|}^{ - 2}}{K^3}}}{{{\pi ^2}\left( {N_r^3 - {N_r}} \right){M_t}{P_{\max }}{N^2}\sum\nolimits_{k = 1}^K {\rho _{{\rm{BI,}}k}^2} }}\nonumber\\
&\sim {\cal{O}}\left( {{K^3}/\left( {{N^2}\left( {N_r^3 - {N_r}} \right)} \right)} \right).
\end{align}
The result in \eqref{CRB_large_N} sheds light on the interplay between the CRB and the number of elements/sensors deployed at IRSs, i.e., $N$ and ${N_r}$, which implies that CRB is inversely proportional to $\left( {N_r^3 - {N_r}} \right){N^2}$. Specifically, the reflecting elements provide a passive beamforming gain of ${\cal{O}}\left( {{N^2}} \right)$, while the sensors not only introduce the linear receive beamforming gain but also enhance the phase difference, i.e., $\left\| {{{{\bf{\dot a}}}_s}\left( {{u_T}} \right)} \right\|_2^2$.
\end{rem}
\subsection{General Case: Separate CU and Target}
For the general case that the locations of the CU and target are separated, the corresponding problem becomes more challenging than that of the co-located CU and target case since the IRSs phase-shifts are coupled with power allocations in both \eqref{C7-a} and \eqref{C7-b}. To tackle this difficulty, an alternating-based SCA algorithm is proposed. In particular, we partition all involved variables into two blocks, i.e., IRSs phase-shifts $\left\{ {{{\bf{\Theta }}_k}} \right\}$ and power allocations $\left\{ {{p_{c,k,}}{p_s}} \right\}$,  and then iteratively optimize these two blocks.
\subsubsection{Optimizing IRSs Phase-shifts}
For any given $\left\{ {{p_{c,k,}}{p_s}} \right\}$, the subproblem of optimizing $\left\{ {{{\bf{\Theta }}_k}} \right\}$ can be written as
\begin{subequations}\label{C100}
\begin{align}
\label{C100-a}\mathop {\max }\limits_{\left\{ {{{\bf{\Theta }}_k}} \right\}}\;\;&\sum\nolimits_{k = 1}^K {{{\log }_2}\left( {1 + \frac{{{p_{c,k}}{{\left| {\Upsilon \left( {{{\bf{\Theta }}_k}} \right)} \right|}^2}}}{{{\sigma ^2}}}} \right)}  \\
\label{C100-b}{\rm{s.t.}}\;\;\;&{f_s}\left( {\left\{ {{p_{c,k}}} \right\},{p_s}} \right) \ge {\Gamma _s},\\
\label{C100-c}&\eqref{C1-e}.
\end{align}
\end{subequations}
To facilitate the optimization, we define the following variables and notations: ${{\bf{v}}_k} = {\mathop{\rm Diag}\nolimits} \left( {{{\bf{\Theta }}_k}} \right) \in {\mathbb{C}^{{N_k} \times 1}}$, ${\bf{q}}_{c,k}^H = {\mathop{\rm diag}\nolimits} \left( {{{\bf{b}}_{{\rm{I}},k}}\left( {\mu _{{\rm{BI}},k}^{\rm{A}}} \right)} \right){\bf{b}}_{{\rm{I}},k}^H\left( {\mu _{{\rm{IU}},k}^{\rm{D}}} \right)\in {\mathbb{C}^{1 \times {N_k}}}$, ${{\bf{Q}}_{c,k}} = {{\bf{q}}_{c,k}}{\bf{q}}_{c,k}^H\in {\mathbb{C}^{{N_k} \times {N_k}}}$, ${\bf{q}}_{s,k}^H = {\mathop{\rm diag}\nolimits} \left( {{{\bf{b}}_{{\rm{I}},k}}\left( {\mu _{{\rm{BI}},k}^{\rm{A}}} \right)} \right){\bf{b}}_{{\rm{I}},k}^H\left( {\mu _{{\rm{IT}},k}^{\rm{D}}} \right)\in {\mathbb{C}^{1 \times {N_k}}}$, and ${{\bf{Q}}_{s,k}} = {{\bf{q}}_{s,k}}{\bf{q}}_{s,k}^H\in {\mathbb{C}^{{N_k} \times {N_k}}}$. With these definitions, the terms ${\left| {\Upsilon \left( {{{\bf{\Theta }}_k}} \right)} \right|^2}$ and ${\left| {\gamma _k^s} \right|^2}$ in \eqref{C7-a} and \eqref{C7-b} can be rewritten as
\begin{align}\label{terms100}
&{\left| {\Upsilon \left( {{{\bf{\Theta }}_k}} \right)} \right|^2} = {M_r}{M_t}\rho _{c,k}^2{\bf{v}}_k^H{{\bf{Q}}_{c,k}}{{\bf{v}}_k},\forall k,\nonumber\\
&{\left| {\gamma _k^s} \right|^2} = {\bf{v}}_k^H{{\bf{Q}}_{s,k}}{{\bf{v}}_k}, \forall k.
\end{align}
By introducing a set of slack variables $\left\{ {{S_k}} \right\}$, problem \eqref{C100}  can be recast as
\begin{subequations}\label{C101}
\begin{align}
\label{C101-a}\mathop {\max }\limits_{\left\{ {{{{\bf{v}}_k}}} \right\}}\;\;&\sum\nolimits_{k = 1}^K {{{\log }_2}\left( {1 + \frac{{{p_{c,k}}{S_k}}}{{{\sigma ^2}}}} \right)} \\
\label{C101-b}{\rm{s.t.}}\;\;\;&{S_k} \le {M_r}{M_t}\rho _{c,k}^2{\bf{v}}_k^H{{\bf{Q}}_{c,k}}{{\bf{v}}_k},\forall k,\\
\label{C101-c}&{M_t}\sum\nolimits_{k = 1}^K {\left( {{p_s} + {p_{c,k}}} \right)\rho _{{\rm{BI}},k}^2{\bf{v}}_k^H{{\bf{Q}}_{s,k}}{{\bf{v}}_k}}  \ge {\Gamma _s},\\
\label{C101-d}&\left| {{{\left[ {{{\bf{v}}_k}} \right]}_n}} \right| = 1,\forall k,\forall n.
\end{align}
\end{subequations}
Note that constraints \eqref{C101-b} and \eqref{C101-c} involve a super-level set of the convex quadratic functions, i.e., ${\bf{v}}_k^H{{\bf{Q}}_{c,k}}{{\bf{v}}_k}$ and ${{\bf{v}}_k^H{{\bf{Q}}_{s,k}}{{\bf{v}}_k}}$, which are non-convex. Additionally, the unit-modulus constraints in \eqref{C100-c} are  inherently non-convex as well. To handle these issues, we first relax the unit-modulus constraints as
\begin{align}\label{relax_unit_modulus}
\left| {{{\left[ {{{\bf{v}}_k}} \right]}_n}} \right| \le 1,\forall k,\forall n.
\end{align}
Regarding constraints \eqref{C101-b} and \eqref{C101-c}, we invoke the SCA technique to linearize ${\bf{v}}_k^H{{\bf{Q}}_{c,k}}{{\bf{v}}_k}$ and ${{\bf{v}}_k^H{{\bf{Q}}_{s,k}}{{\bf{v}}_k}}$  into linear forms by taking their first-order Taylor expansion at the given point ${\bf{v}}_k^{\left( l \right)}$ in the $l$-th iteration, yielding the following inequality
\begin{align}\label{SCA_inquality1}
{\bf{v}}_k^H{{\bf{Q}}_{s,k}}{{\bf{v}}_k} &\ge  - {\left( {{\bf{v}}_k^{\left( l \right)}} \right)^H}{{\bf{Q}}_{s,k}}{\bf{v}}_k^{\left( l \right)} + 2{\mathop{\rm Re}\nolimits} \left( {{{\left( {{\bf{v}}_k^{\left( l \right)}} \right)}^H}{{\bf{Q}}_{s,k}}{\bf{v}}} \right)\nonumber\\
&\buildrel \Delta \over = g_s^{{\rm{lb}}}\left( {{{\bf{v}}_k}} \right),
\end{align}
\vspace{-17pt}
\begin{align}\label{SCA_inquality2}
{\bf{v}}_k^H{{\bf{Q}}_{c,k}}{{\bf{v}}_k} &\ge  - {\left( {{\bf{v}}_k^{\left( l \right)}} \right)^H}{{\bf{Q}}_{c,k}}{\bf{v}}_k^{\left( l \right)} + 2{\mathop{\rm Re}\nolimits} \left( {{{\left( {{\bf{v}}_k^{\left( l \right)}} \right)}^H}{{\bf{Q}}_{c,k}}{\bf{v}}} \right)\nonumber\\
&\buildrel \Delta \over = g_c^{{\rm{lb}}}\left( {{{\bf{v}}_k}} \right).
\end{align}
Based on \eqref{SCA_inquality1} and \eqref{SCA_inquality2}, the newly formulated problem is obtained as
\begin{subequations}\label{C102}
\begin{align}
\label{C102-a}\mathop {\max }\limits_{\left\{ {{{{\bf{v}}_k}}} \right\}}\;\;&\sum\nolimits_{k = 1}^K {{{\log }_2}\left( {1 + \frac{{{p_{c,k}}{S_k}}}{{{\sigma ^2}}}} \right)} \\
\label{C102-b}{\rm{s.t.}}\;\;\;&{S_k} \le {M_r}{M_t}\rho _{c,k}^2g_c^{{\rm{lb}}}\left( {{{\bf{v}}_k}} \right),\forall k,\\
\label{C102-c}&{M_t}\sum\nolimits_{k = 1}^K {\left( {{p_s} + {p_{c,k}}} \right)\rho _{{\rm{BI}},k}^2g_s^{{\rm{lb}}}\left( {{{\bf{v}}_k}} \right)}  \ge {\Gamma _s},\\
\label{C102-c}&\left| {{{\left[ {{{\bf{v}}_k}} \right]}_n}} \right| \le 1,\forall k,\forall n.
\end{align}
\end{subequations}
It can be readily checked problem \eqref{C102} is convex, which thus can be solved by the interior point method.
\subsubsection{Optimizing Power Allocations}
Under the given $\left\{ {{{\bf{\Theta }}_k}} \right\}$, the optimization problem with respect to $\left\{ {{p_{c,k,}}{p_s}} \right\}$ is given by
\begin{subequations}\label{C103}
\begin{align}
\label{C103-a}\mathop {\max }\limits_{\left\{ {{p_{c,k}},{p_s}} \right\}}\;\;&\sum\nolimits_{k = 1}^K {{{\log }_2}\left( {1 + \frac{{{p_{c,k}}{\Upsilon ^2}\left( {{{\bf{\Theta }}_k}} \right)}}{{{\sigma ^2}}}} \right)} \\
\label{C103-b}&\eqref{C7-b},\eqref{C7-c},\eqref{C7-d}.
\end{align}
\end{subequations}
Problem \eqref{C103} is convex and its optimal solution can be derived similarly in Theorem 2, which is omitted for brevity.

\subsubsection{Overall Algorithm}
Summarily, the proposed algorithm updates $\left\{ {{{\bf{v}}_k}} \right\}$ and $\left\{ {{p_{c,k,}}{p_s}} \right\}$ alternately in an iterative manner, where the former is updated by solving subproblem \eqref{C102}, and the latter is updated by solving subproblem \eqref{C103}. The iteration terminates until the convergence is achieved. However, it is worth noting that the converged solution, denoted by $\left\{ {{{{\bf{\bar v}}}_k}} \right\}$, may not satisfy the unit-modules constraints in \eqref{C101-d}. In this case, one feasible solution, denoted by $\left\{ {{{{\bf{\tilde v}}}_k}} \right\}$, can be reconstructed as ${\left[ {{{{\bf{\tilde v}}}_k}} \right]_n} = {\left[ {{{{\bf{\bar v}}}_k}} \right]_n}/\left| {{{\left[ {{{{\bf{\bar v}}}_k}} \right]}_n}} \right|,\forall k,n$. With given $\left\{ {{{{\bf{\tilde v}}}_k}} \right\}$, the remaining optimization variables can be obtained by solving problem \eqref{C103}. The main complexity of the proposed algorithm is given by ${\cal{O}}\left( {{L_{{\rm{iter}}}}\left( {{{\left( {N + K} \right)}^{3.5}} + {{\left( {K + 1} \right)}^{3.5}}} \right)} \right)$, where ${{L_{{\rm{iter}}}}}$ denote the total number of iterations required for achieving convergence.

\section{Numerical results}
In this section, numerical results are provided to validate the performance of the proposed design and to draw useful insights. The BS and CU/target are located at $\left( {0,{\rm{ }}0} \right)$ meter (m) and $\left( {100,{\rm{ }}0} \right)$ m, respectively. The number of the transmit antennas at the BS and the receive antennas at the CU are set to $32$ and $8$, respectively. Under this setup, the candidate set of the potential IRS sites $\tilde {\cal P}_{{\rm{IRS}}}^{\rm{c}}$ can be obtained, which satisfies $\left| {\widetilde {\cal P}_{{\rm{IRS}}}^{\rm{c}}} \right| = \min \left\{ {{M_t},{M_r}} \right\} = 8$. The top view of the coordinates of BS, CU/target, and eight potential IRS sites is illustrated in Fig. \ref{location}. In Fig. \ref{location}, IRS 1 refers to the S-IRS, which is also equipped with active sensors. We adopt the distance-based path-loss model, i.e., $L\left( {d,\alpha } \right) = {K_0}{d^{ - \alpha }}$, with ${K_0} =  - 40$ dB and $\alpha  = 2$. Without loss of generality, other system parameters are set as follows: $T = 256$, ${P_{\max }} = 30$ dBm, ${N_r} = 8$, ${\sigma ^2} =  - 80$ dBm, and $\sigma _s^2 =  - 110$ dBm.
In the following, we first show the impact of the number of deployed IRSs on the sensing performance and then show the communication-sensing tradeoff under the proposed multi-IRS aided ISAC system.
\begin{figure}[t!]
\centering
\includegraphics[width=2.7in]{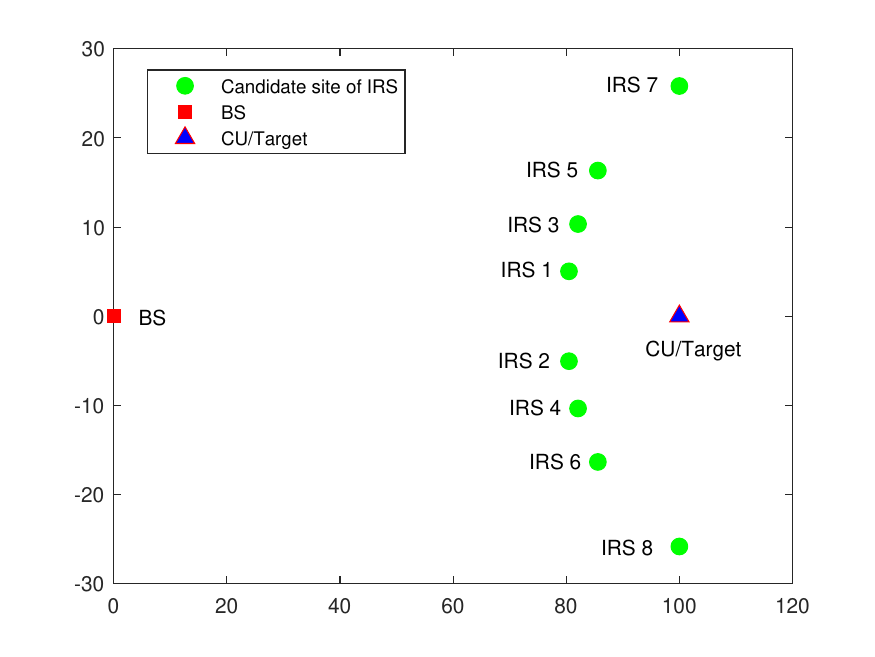}
\caption{{Top view of the coordinates of BS, CU/Target, and IRSs.}}
\label{location}
\end{figure}

\subsection{Multi-IRS Aided Sensing}
We first investigate a special case of ISAC, where the BS only provides sensing service to the target. For comparison, the following transmission schemes are considered: 1) \textbf{Sensing-oriented  transmission}: the IRS phase-shifts and transmit covariance matrix are set based on \eqref{phase_IRS_opt_sensing} and \eqref{transmit_opt_sensing}, respectively, to minimize CRB; 2) \textbf{Communication-oriented transmission}: the IRS phase-shifts and transmit covariance matrix are set based on \eqref{phase_IRS_opt_com} and \eqref{transmit_opt_com}, respectively, to maximize the rate of CU; 3)\textbf{Maximum eigenmode transmission}: the IRS phase-shifts are optimized based on \eqref{phase_IRS_opt_sensing} and the transmit covariance matrix is set to ${{\bf{R}}_x} = {P_{\max }}{{\bf{a}}_{\rm{B}}}\left( {\mu _{{\rm{BI,}}1}^{\rm{D}}} \right){\bf{a}}_{\rm{B}}^H\left( {\mu _{{\rm{BI,}}1}^{\rm{D}}} \right)/\left\| {{{\bf{a}}_B}\left( {\mu _{{\rm{BI,}}1}^{\rm{D}}} \right)} \right\|_2^2$.

\begin{figure}[t!]
\centering
\includegraphics[width=2.8in]{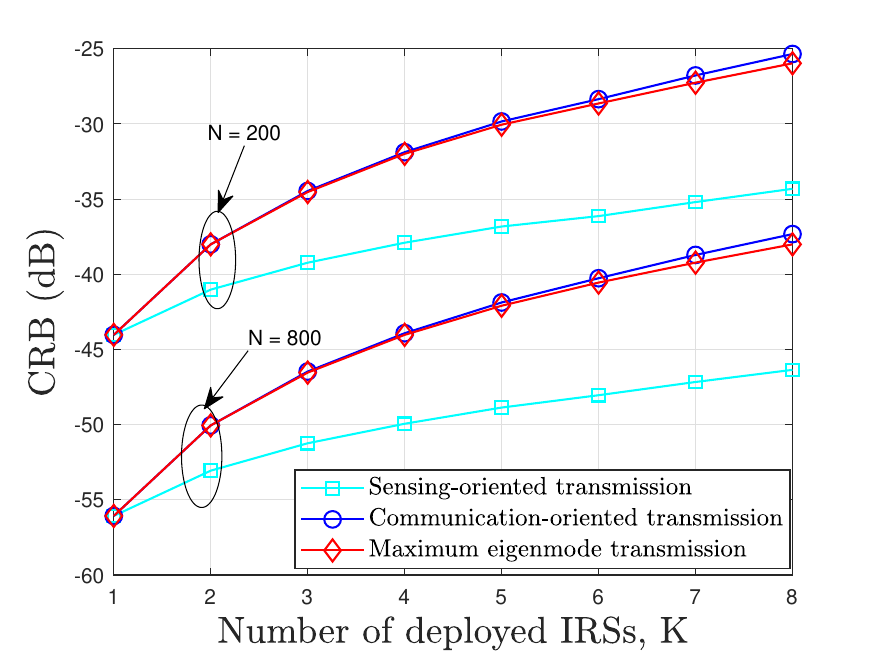}
\caption{{The CRB versus the number of IRSs $K$ under the given total number of passive elements.}}
\label{sensing1}
\end{figure}

In Fig. \ref{sensing1}, we plot the estimation CRB versus the number of deployed IRSs $K$ under the given total number of IRS elements $N$. First, it is observed that the CRB achieved by the sensing-oriented transmission is significantly lower than that achieved by the other two benchmark schemes especially as $K$ becomes large, which validates the effectiveness of the proposed design. Moreover, the estimation CRB of all the considered schemes increases as $K$ increases under the given $N$. This is because increasing the number of IRSs leads to more paths for the BS-IRS link, which brings difficulty in aligning the signals from different paths to the target and thus undermines the echo power at sensors. In contrast, deploying a single large-scale IRS is beneficial for aligning the signal to the target, thereby improving the echo power. The result demonstrates that increasing $K$ is harmful to the sensing performance, which agrees with the analysis in Proposition 1.

\begin{figure}[t!]
\centering
\includegraphics[width=2.8in]{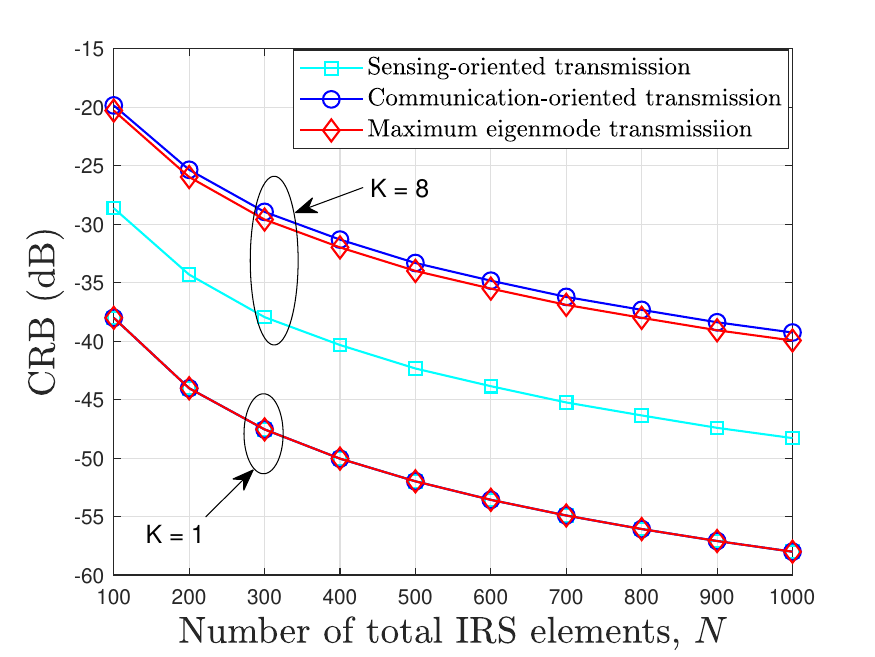}
\caption{{The CRB versus the number of total IRS elements under the different values of $K$.}}
\label{sensing1}
\end{figure}

To investigate the impact of the number of total IRS elements $N$ on the sensing performance, we plot the sensing CRB versus $N$ under the different numbers of deployed IRSs. It is observed that all the considered schemes achieve the same CRB under the setup of $K = 1$. The result is expected since both the schemes of communication-oriented transmission and maximum eigenmode transmission reduce to the sensing-oriented transmission when $K = 1$. Furthermore, one can observe that the CRB reductions of all the schemes are about 20 dB by increasing $N$ from 100 to 1000, which agrees with the discussion in Remark 3 that CRB is inversely proportional to ${N^2}$.

\begin{figure}[t!]
\centering
\includegraphics[width=2.8in]{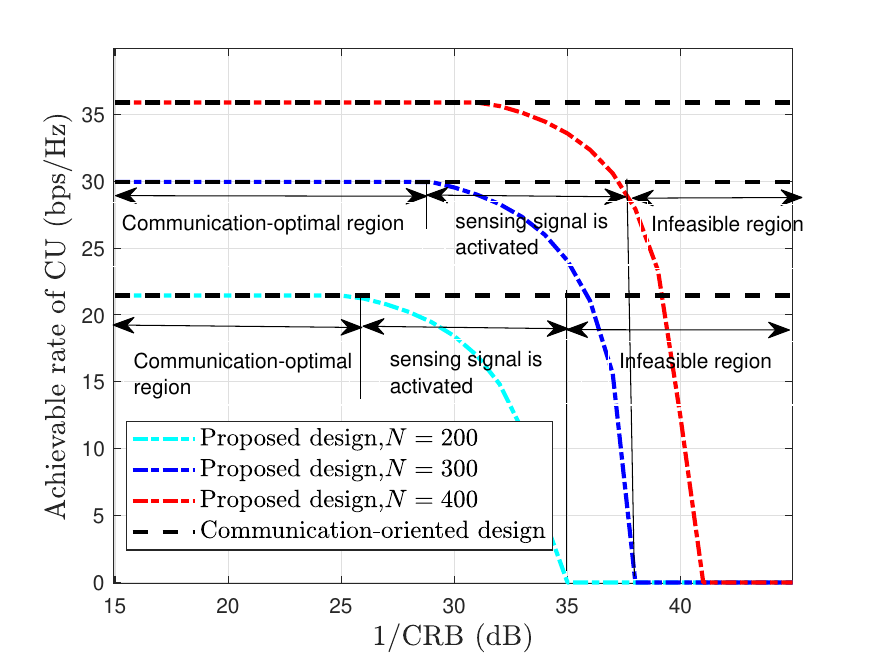}
\caption{{Achievable rate versus the multiplicative inverse of the required estimation CRB with $K = 8$.}}
\label{tradeoff1}
\end{figure}

\subsection{Multi-IRS Aided ISAC}
Next, we consider the case where the BS provides both the sensing and communication service in a multi-IRS-aided ISAC system. To illustrate the fundamental tradeoff between the communication and sensing performance, we plot the achievable rate of CU versus the multiplicative inverse of the required estimation CRB in Fig. \ref{tradeoff1}. Note that the communication-oriented design serves as the upper bound of the communication rate. It is observed that the proposed design can achieve the upper bound of the communication rate when the requirement of sensing performance is loose, wherein the dedicated sensing signal is not needed. As the requirement CRB decreases (i.e., $1/{\mathop{\rm CRB}\nolimits}$ increases), it can be observed that there exists a tradeoff between the achievable rate and estimation CRB, wherein the dedicated sensing signal is activated to meet the requirement of sensing performance. Nevertheless, it is worth noting that the region for ensuring the optimality of communication-oriented design can be effectively enlarged by increasing the total number of IRS elements $N$. The result demonstrates the potential of deploying large IRSs in achieving high-quality services of both sensing and communication.

\begin{figure}[t!]
\centering
\includegraphics[width=2.8in]{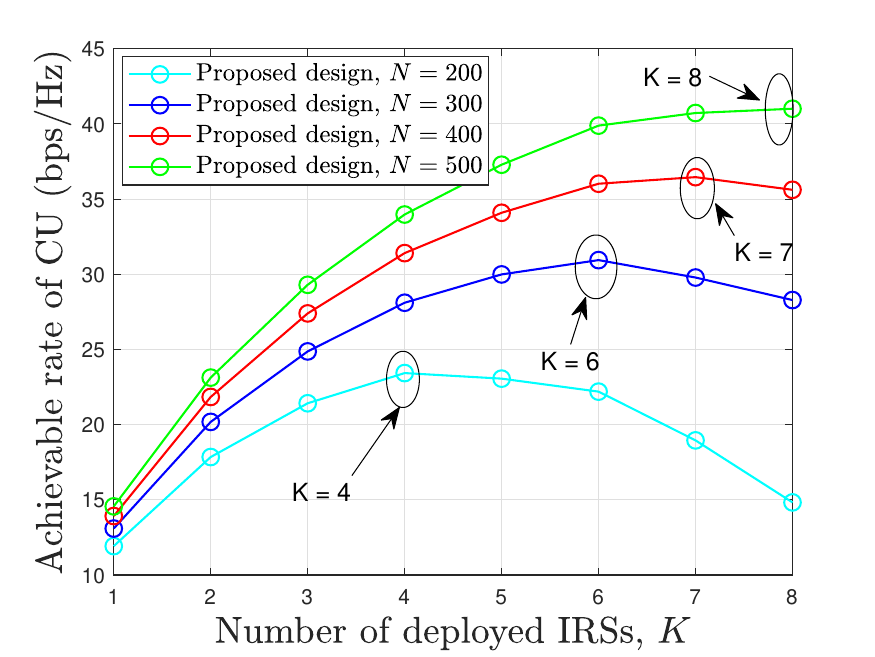}
\caption{{Achievable rate versus the number of deployed IRSs with the required CRB $-32$ dB.}}
\label{deployment_dis}
\end{figure}

In Fig. \ref{deployment_dis}, we study the impact of IRSs deployment on the system performance by plotting the achievable rate of CU versus the number of deployed IRSs under the given sensing performance requirement. Under the low value of $N$, it is observed that the achievable rate first increases and then decreases with $K$ increases, which indicates that there exists an optimal number of deployed IRSs. The result highlights the importance of optimizing the set ${{\cal P}_{{\rm{IRS}}}}$ to flexibly balance the sensing and communication performance. In contrast, we observe that the achievable rate of CU monotonically increases with $K$ increases when $N$ becomes large. The reason is that the power loss incurred by the multi-path of the first-hop BS-IRS link can be effectively compensated by the high passive beamforming gain introduced by large-scale IRSs. In this case, increasing the number of IRSs significantly improve the DoFs of spatial multiplexing, which is able to enhance the communication rate while satisfying the sensing performance requirement.

\begin{figure}[t!]
\centering
\includegraphics[width=2.8in]{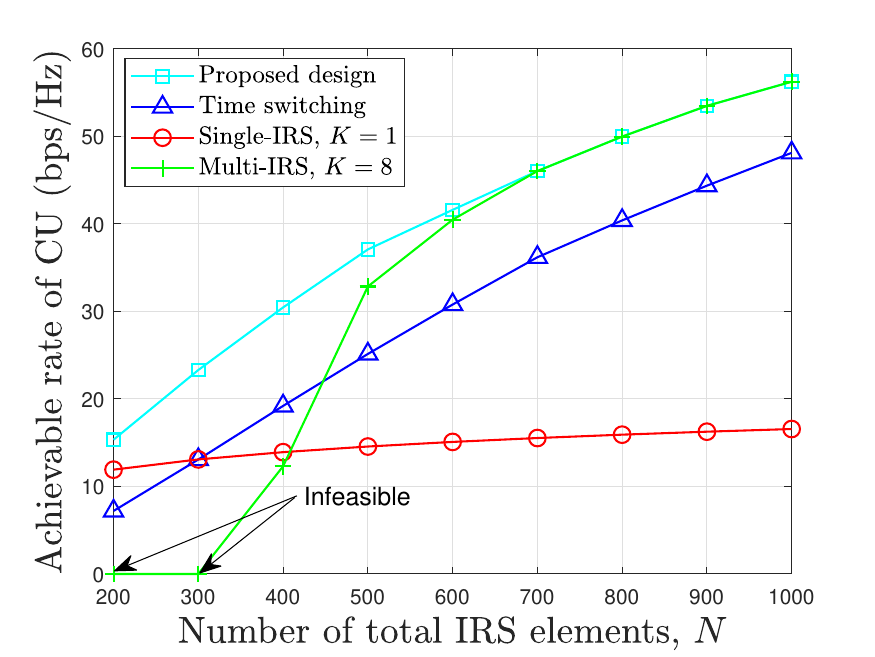}
\caption{{Achievable rate versus the number of total IRS elements with the target CRB $-42$ dB.}}
\label{performance1}
\end{figure}

Finally, we evaluate the performance of the proposed designs in terms of the joint optimization of the transceiver designs and IRSs deployment. For comparison, we consider the following benchmark schemes: 1) \textbf{Time-switching}: Under the optimized IRSs position site ${{\cal P}_{{\rm{IRS}}}}$, the BS time switches between the communication-oriented design and sensing-oriented design; 2) \textbf{Single-IRS}, $K = 1$: the transceiver design in Section IV is performed under ${{\cal P}_{{\rm{IRS}}}} = \left\{ {{\rm{IRS ~1}}} \right\}$; 3) \textbf{Multi-IRS, $K = 8$}: the transceiver design in Section IV is performed under ${{\cal P}_{{\rm{IRS}}}} = \tilde {\cal P}_{{\rm{IRS}}}^{\rm{c}}$. For all the considered schemes, we plot the achievable rate of the CU versus $N$ in Fig. \ref{performance1}.

From Fig. \ref{performance1}, it is observed that our proposed design outperforms other benchmark schemes, which demonstrates the importance of the joint optimization of IRS deployment and transceiver designs. Additionally, the performance of the scheme with single-IRS is less sensitive to the increase of $N$ and outperforms the time-switching scheme for small $N$, whereas the time-switching scheme outperforms the scheme with single-IRS for large $N$. The reason is that the DoF of spatial multiplexing in the single-IRS case is limited. Moreover, it is observed that the scheme of $K = 8$ cannot meet the sensing requirement for small $N$ due to the loss of passive beamforming gain. Nevertheless, the achievable rate of $K = 8$ increases significantly as $N$ increases and achieves the same performance as that of the proposed design for large $N$. The result is expected since the DoFs of spatial multiplexing are fully unlocked by deploying large $N$, whereas the sensing performance requirement can be also met due to a high passive beamforming gain of large $N$.

\vspace{-2pt}
\section{Conclusion}
This paper studied the fundamental tradeoff between sensing and communication in a hybrid multi-IRS-aided ISAC system from the viewpoint of multi-path exploitation versus reduction. First, the impact of the number of deployed IRSs on both the sensing and communication performance was analyzed, which demonstrates that increasing the number of deployed IRSs is beneficial for improving DoFs of spatial multiplexing for communication while increasing the CRB of target estimation. To characterize this tradeoff, we studied a rate maximization problem, by optimizing the BS transmit covariance matrix, IRS phase-shifts, and the number of deployed IRSs, subject to a maximum CRB constraint. Analytical results unveiled that the communication-oriented design becomes optimal when the total number of IRS elements exceeds a certain threshold. This provides useful guidelines for the transceiver design in ISAC systems with large IRSs. Finally, numerical results were provided to validate the theoretical findings.

\section*{Appendix A: \textsc{Proof of Lemma 1}}
When $\left| {\mu _{{\rm{BI}},k}^{\rm{D}} - \mu _{{\rm{BI}},i}^{\rm{D}}} \right| = \frac{{2m}}{{{M_t}}}$ holds, we have
\begin{align}\label{correlation}
\begin{array}{*{20}{l}}
{{\bf{\tilde a}}_{\rm{B}}^H\left( {\mu _{{\rm{BI}},k}^{\rm{D}}} \right){{{\bf{\tilde a}}}_{\rm{B}}}\left( {\mu _{{\rm{BI}},i}^{\rm{D}}} \right)}\\
{ = \frac{{{e^{ - j\frac{{\left( {{M_t} - 1} \right)\pi \left( {\mu _{{\rm{BI}},{\rm{k}}}^{\rm{D}}{\rm{ - }}\mu _{{\rm{BI}},{\rm{i}}}^{\rm{D}}} \right)}}{2}\sum\limits_{m = 0}^{{M_t} - 1} {{e^{j\pi m\left( {\mu _{{\rm{BI}},{\rm{k}}}^{\rm{D}}{\rm{ - }}\mu _{{\rm{BI}},{\rm{i}}}^{\rm{D}}} \right)}}} }}}}{{\sqrt {{M_t}{M_r}} }}}\\
{ = \frac{{{e^{ - j\frac{{\left( {{M_t} - 1} \right)\pi \left( {\mu _{{\rm{BI}},{\rm{k}}}^{\rm{D}}{\rm{ - }}\mu _{{\rm{BI}},{\rm{i}}}^{\rm{D}}} \right)}}{2}}}}}{{\sqrt {{M_t}{M_r}} }}\frac{{1 - {e^{j2m\pi }}}}{{1 - {e^{j\left( {\mu _{{\rm{BI}},{\rm{k}}}^{\rm{D}}{\rm{ - }}\mu _{{\rm{BI}},{\rm{i}}}^{\rm{D}}} \right)}}}}}\\
{ = 0,\forall k \ne i.}
\end{array}
\end{align}
It is obvious that ${\bf{\tilde a}}_{\rm{B}}^H\left( {\mu _{{\rm{BI}},k}^{\rm{D}}} \right){{{\bf{\tilde a}}}_{\rm{B}}}\left( {\mu _{{\rm{BI}},k}^{\rm{D}}} \right)= 1$, which leads to ${\bf{V}}_{\rm{c}}^H{{\bf{V}}_{\rm{c}}} = {{\bf{I}}_K}$.
Similarly, it can be shown that ${\bf{U}}_{\rm{c}}^H{{\bf{U}}_{\rm{c}}} = {{\bf{I}}_K}$  due to its brevity.

\section*{Appendix B: \textsc{Proof of Theorem 1}}
To derive the CRB of the ${{\mu _{\rm{T}}}}$ estimation, we rewrite \eqref{received_signals_sensors_modify} as follows
\begin{align}\label{received_signals_sensors_modify1}
{{{\bf{\bar y}}}_s}\left( t \right) &\mathop  = \limits^{\left( a \right)}  \beta {\rho _{{\rm{TS}}}}{\rho _{{\rm{IT}},1}}{{\bf{a}}_s}\left( {{\mu _{\rm{T}}}} \right){\bf{b}}_{{\rm{I}},1}^H\left( {{\mu _{\rm{T}}}} \right){{\bf{\Theta }}_1}{{\bf{G}}_1}{\bf{x}}\left( t \right)\nonumber\\
& + {{\bf{a}}_s}\left( {{\mu _{\rm{T}}}} \right)\sum\limits_{k = 2}^K {\beta {\rho _{{\rm{TS}}}}{\rho _{{\rm{IT}},k}}{\bf{b}}_{{\rm{I}},k}^H\left( {\mu _{{\rm{IT}},k}^{\rm{D}}} \right){{\bf{\Theta }}_k}{{\bf{G}}_k}{\bf{x}}\left( t \right)}  \!\!+\!\! {{\bf{n}}_s}\left( t \right)\nonumber\\
&\approx \tilde \beta {{\bf{a}}_s}\left( {{\mu _{\rm{T}}}} \right){\bf{b}}_{{\rm{I}},1}^H\left( {{\mu _{\rm{T}}}} \right){{\bf{\Theta }}_1}{{\bf{G}}_1}{\bf{x}}\left( t \right)\nonumber\\
& + \tilde \beta {{\bf{a}}_s}\left( {{\mu _{\rm{T}}}} \right)\sum\limits_{k = 2}^K {{\bf{b}}_{{\rm{I}},k}^H\left( {\mu _{{\rm{IT}},k}^{\rm{D}}} \right){{\bf{\Theta }}_k}{{\bf{G}}_k}{\bf{x}}\left( t \right)} \!\!+\!\! {{\bf{n}}_s}\left( t \right),
\end{align}
where ${\left( a \right)}$ holds since the target is in the far field of the semi-passive IRS, i.e., ${\mu _{\rm{T}}} = \mu _{{\rm{IT}},1}^{\rm{D}}$ and ${\left| {\tilde \beta } \right|^2} = \max \left\{ {{{\left| \beta  \right|}^2}\rho _{{\rm{TS}}}^2\rho _{{\rm{IT}},1}^2, \ldots {{\left| \beta  \right|}^2}\rho _{{\rm{TS}}}^2\rho _{{\rm{IT}},K}^2} \right\}$. Let
\begin{align}\label{temp}
{\bf{A}}\left( {{\mu _{\rm{T}}}} \right) =& {{\bf{a}}_s}\left( {{\mu _{\rm{T}}}} \right){\bf{b}}_{{\rm{I}},1}^H\left( {{\mu _{\rm{T}}}} \right){{\bf{\Theta }}_1}{{\bf{G}}_1}\nonumber\\
&+ {{\bf{a}}_s}\left( {{\mu _{\rm{T}}}} \right)\sum\nolimits_{k = 2}^K {{\bf{b}}_{{\rm{I}},k}^H\left( {\mu _{{\rm{IT}},k}^{\rm{D}}} \right){{\bf{\Theta }}_k}{{\bf{G}}_k}}
\end{align}
and we vectorize ${{{\bf{\bar Y}}}_s} = \left[ {{{{\bf{\bar y}}}_s}\left( 1 \right), \ldots ,{{{\bf{\bar y}}}_s}\left( T \right)} \right]$ as
\begin{align}\label{sensing_signal_appro}
{\mathop{\rm vec}\nolimits} \left( {{{{\bf{\bar Y}}}_s}} \right) \approx \tilde \beta {\mathop{\rm vec}\nolimits} \left( {{\bf{A}}\left( {{\mu _{\rm{T}}}} \right){\bf{X}}} \right) + {{\bf{N}}_s},
\end{align}
where ${\bf{X}} = \left[ {{\bf{x}}\left( 1 \right), \ldots ,{\bf{x}}\left( T \right)} \right]$ and ${{\bf{N}}_s} = \left[ {{{\bf{n}}_s}\left( 1 \right), \ldots ,{{\bf{n}}_s}\left( T \right)} \right]$.

Since we focus on characterizing the performance of estimating the angle, we define ${{{\bm{\tilde \beta }}}_s} = {\left[ {{\mathop{\rm Re}\nolimits} \left\{ {\tilde \beta } \right\},{\mathop{\rm Im}\nolimits} \left\{ {\tilde \beta } \right\}} \right]^T}$. Let ${\bm{\xi }} = {\left[ {{\mu _{\rm{T}}},{{{\bf{\tilde \beta }}}_s}^T} \right]^T}$ and then the Fisher information matrix (FIM) with respect to the estimated parameters ${\bm{\xi }}$ is given by
\begin{align}\label{FIM}
{\bf{F}} = \left[ {\begin{array}{*{20}{c}}
{{{\bf{F}}_{{\mu _{\rm{T}}},{\mu _{\rm{T}}}}}}&{{{\bf{F}}_{{\mu _{\rm{T}}},{{{\bm{\tilde \beta }}}_s}}}}\\
{{\bf{F}}_{{\mu _{\rm{T}}},{{{\bm{\tilde \beta }}}_s}}^T}&{{{\bf{F}}_{{{{\bm{\tilde \beta }}}_s},{{{\bm{\tilde \beta }}}_s}}}}
\end{array}} \right],
\end{align}
where each element in ${\bf{F}}\left( {\bf{\xi }} \right)$ is defined as
\begin{align}\label{FIM_element}
{{\bf{F}}_{l,m}} = \frac{2}{{\sigma _s^2}}{\mathop{\rm Re}\nolimits} \left\{ {\frac{{\partial {{{\bm{\tilde u}}}^H}}}{{\partial {{\bm{\xi }}_i}}}\frac{{\partial {\bm{\tilde u}}}}{{\partial {{\bm{\xi }}_j}}}} \right\},l,m \in \left\{ {1,2,3} \right\}
\end{align}
with ${\bf{\tilde u}} = \tilde \beta {\rm{vec}}\left( {{\bf{A}}\left( {{\mu _{\rm{T}}}} \right){\bf{X}}} \right)$. Based on the FIM, the CRB of the estimated ${{\mu _{\rm{T}}}}$ is given by
\begin{align}\label{CRB_theta_def}
{\mathop{\rm CRB}\nolimits} \left( {{\mu _{\rm{T}}}} \right) = {\left[ {{{\bf{F}}_{{\mu _{\rm{T}}},{\mu _{\rm{T}}}}} - {{\bf{F}}_{{\mu _{\rm{T}}},{{{\bm{\tilde \beta }}}_s}}}{\bf{F}}_{{{{\bm{\tilde \beta }}}_s},{{{\bm{\tilde \beta }}}_s}}^{ - 1}{\bf{F}}_{{\mu _{\rm{T}}},{{{\bm{\tilde \beta }}}_s}}^T} \right]^{ - 1}}.
\end{align}
In \eqref{CRB_theta_def}, ${{{\bf{F}}_{{\mu _{\rm{T}}},{\mu _{\rm{T}}}}}}$, ${{{\bf{F}}_{{\mu _{\rm{T}}},{{{\bm{\tilde \beta }}}_s}}}}$, and ${{{\bf{F}}_{{{{\bm{\tilde \beta }}}_s},{{{\bm{\tilde \beta }}}_s}}}}$ can be calculated as
\begin{align}\label{entry1}
&{{\bf{F}}_{{\mu _{\rm{T}}},{\mu _{\rm{T}}}}} = \frac{{2T{{\left| {\tilde \beta } \right|}^2}}}{{\sigma _s^2}}{\mathop{\rm Tr}\nolimits} \left( {{\bf{\dot A}}\left( {{\mu _{\rm{T}}}} \right){{\bf{R}}_x}{{{\bf{\dot A}}}^H}\left( {{\mu _{\rm{T}}}} \right)} \right)\nonumber\\
&{{\bf{F}}_{{\mu _{\rm{T}}},{{{\bm{\tilde \beta }}}_s}}} = \frac{{2T}}{{\sigma _s^2}}{\mathop{\rm Re}\nolimits} \left\{ {{{\tilde \beta }^H}{\mathop{\rm Tr}\nolimits} \left( {{\bf{A}}\left( {{\mu _{\rm{T}}}} \right){{\bf{R}}_x}{{{\bf{\dot A}}}^H}\left( {{\mu _{\rm{T}}}} \right)} \right)\left( {1,j} \right)} \right\},\nonumber\\
&{{\bf{F}}_{{{{\bm{\tilde \beta }}}_s},{{{\bm{\tilde \beta }}}_s}}} = \frac{{2T}}{{\sigma _s^2}}{\rm{Tr}}\left( {{\bf{A}}\left( {{\mu _{\rm{T}}}} \right){{\bf{R}}_x}{{\bf{A}}^H}\left( {{\mu _{\rm{T}}}} \right)} \right){{\bf{I}}_2},
\end{align}
where
\begin{align}\label{med_term}
{\bf{\dot A}}\left( {{\mu _{\rm{T}}}} \right) =&{{{\bf{\dot a}}}_s}\left( {{\mu _{\rm{T}}}} \right)\sum\nolimits_{k = 1}^K {{\rho _{{\rm{BI}},k}}\gamma _k^s{\bf{a}}_{\rm{B}}^H\left( {\mu _{{\rm{BI}},k}^{\rm{D}}} \right)}\nonumber\\
&+ {{\bf{a}}_s}\left( {{\mu _{\rm{T}}}} \right){\rho _{{\rm{BI}},1}}\tilde \gamma _1^s{\bf{a}}_{\rm{B}}^H\left( {\mu _{{\rm{BI}},1}^{\rm{D}}} \right),
\end{align}
with $\gamma _k^s = {\bf{b}}_{{\rm{I}},k}^H\left( {\mu _{{\rm{IT}},k}^{\rm{D}}} \right){{\bf{\Theta }}_k}{{\bf{b}}_{{\rm{I}},k}}\left( {\mu _{{\rm{BI}},k}^{\rm{A}}} \right),\forall k \in {\cal K}$, and $\tilde \gamma _1^s = {\bf{\dot b}}_{{\rm{I}},1}^H\left( {{\mu _{\rm{T}}}} \right){{\bf{\Theta }}_k}{{\bf{b}}_{{\rm{I}},1}}\left( {\mu _{{\rm{BI}},1}^{\rm{A}}} \right)$. Finally, by plugging \eqref{med_term} into \eqref{CRB_theta_def}, the CRB of the estimated ${{\mu _{\rm{T}}}}$ is obtained in Theorem 1.

\section*{Appendix C: \textsc{Proof of Theorem 2}}
Since problem \eqref{C8} is convex and also satisfies Slater's condition, the dual gap between \eqref{C8} and its dual problem is zero. This demonstrates that the optimal solution can be derived by analyzing  its Karush-Kuhn-Tucker (KKT) conditions. We write the partial Lagrangian function of \eqref{C8} as
\begin{align}\label{Lagrangian_function}
{\cal L} =& \sum\limits_{k = 1}^K {{{\log }_2}\left( {1 + \frac{{{M_t}{M_r}{p_{c,k}}\rho _{{\rm{BI}},k}^2\rho _{{\rm{IU}},k}^2N_k^2}}{{{\sigma ^2}}}} \right)}\nonumber\\
&  + \mu \left( {{p_s}{M_t}\sum\limits_{k = 1}^K {\rho _{{\rm{BI}},k}^2N_k^2}  + {M_t}\sum\limits_{k = 1}^K {\rho _{{\rm{BI}},k}^2N_k^2} {p_{c,k}} - {\Gamma _s}} \right) \nonumber\\
& + \lambda \left( {{P_{\max }} - \sum\limits_{k = 1}^K {{p_{c,k}} - {p_s}} } \right) + \nu {p_s},
\end{align}
where $\mu$, $\lambda$, and $\nu $ denote the Lagrange multipliers associated with constraints \eqref{C8-b}, \eqref{C7-c}, and ${p_s} \ge 0$, respectively. Denote the optimal solution of problem \eqref{C8} and its dual problem as $\left\{ {{p_{c,k}^ \star},{p_s^ \star}} \right\}$ and $\left\{ {\mu^\star ,\lambda^ \star, \nu^ \star } \right\}$, respectively. Following the complementary slackness condition, ${p_s^ \star}\nu^ \star  = 0$ must hold. By taking the partial derivative of \eqref{Lagrangian_function} with respect to ${\left\{ {{p_{c,k}}} \right\}}$ and ${{p_s}}$, respectively, we obtain
\begin{align}\label{KKT_condition1}
\frac{{\partial {\cal L}}}{{\partial {p_{c,k}}}} =&\frac{1}{{\ln 2}}\frac{{{M_t}{M_r}\rho _{{\rm{BI}},k}^2\rho _{{\rm{IU}},k}^2N_k^2}}{{{\sigma ^2} + {M_t}{M_r}{p_{c,k}}\rho _{{\rm{BI}},k}^2\rho _{{\rm{IU}},k}^2N_k^2}}\nonumber\\
&  + \mu {M_t}\rho _{{\rm{BI}},k}^2N_k^2 - \lambda,
\end{align}
\begin{align}\label{KKT_condition2}
\hspace{-0.6cm}\frac{{\partial {\cal L}}}{{\partial {p_s}}} = \mu {M_t}\sum\nolimits_{k = 1}^K {\rho _{{\rm{BI}},k}^2N_k^2}  - \lambda  + \nu.
\end{align}
Then, we derive the optimal solution of problem \eqref{C8} by considering the cases of $p_s^ \star  > 0$ and $p_s^ \star  = 0$, respectively.

For the case that the sensing signal is needed at the optimal solution, i.e., $p_s^ \star  > 0$, we have $\nu^ \star  = 0$ according to ${p_s^ \star}\nu^ \star  = 0$. In this case, $\frac{{\partial {\cal L}}}{{\partial {p_s}}} = 0$, which yields
\begin{align}\label{dual_varialbes_relation}
\lambda  = \mu {M_t}\sum\limits_{k = 1}^K {\rho _{{\rm{BI}},k}^2N_k^2}.
\end{align}
By substituting \eqref{dual_varialbes_relation} into \eqref{KKT_condition1}, we obtain
\begin{align}\label{power_com1}
\hspace{-0.2cm}p_{c,k}^ \star \left( \mu  \right) \!=\! {\left[ {\frac{1}{{\mu {M_t}\sum\nolimits_{l \ne k}^K {\rho _{{\rm{BI}},l}^2N_l^2} }} \!-\! \frac{{{\sigma ^2}}}{{{M_t}{M_r}\rho _{{\rm{BI}},k}^2\rho _{{\rm{IU}},k}^2N_k^2}}} \right]^ + }.
\end{align}
It is obvious that constraint \eqref{C7-c} mets with equality at the optimal solution since otherwise the objective value can be increased by increasing certain ${p_{c,k}}$'s while satisfies the maximum transmit power constraint. As a result, we have ${p_s^\star} = {P_{\max }} - \sum\nolimits_{k = 1}^K p_{c,k}^ \star \left( \mu  \right)$. Since constraint \eqref{C8-b} is tight at the optimal solution, we obtain
\begin{align}\label{dual_varialbes_relation2}
{{\cal G}_{{\rm{s,1}}}}\left( \mu  \right)  \buildrel \Delta \over = & {M_t}\sum\nolimits_{k = 1}^K {\left( {{P_{\max }} - \sum\nolimits_{l \ne k}^K {p_{c,l}^ \star \left( \mu  \right)} } \right){{\left| {{\rho _{{\rm{BI}},k}}} \right|}^2}N_k^2}\nonumber\\
& - {\Gamma _s} = 0.
\end{align}
It is observed from \eqref{power_com1} that $p_{c,k}^ \star \left( \mu  \right)$ decreases with respect to $\mu$, which leads to ${{\cal G}_{{\rm{s,1}}}}\left( \mu  \right)$ is an increasing function of $\mu$. As such, $\mu^ \star$ is the unique solution of equation \eqref{dual_varialbes_relation2}, which can be obtained by using a bisection search. Hence, the condition
\begin{align}\label{sensing_signal_condition}
&\sum\nolimits_{k = 1}^K {{{\left[ {\frac{1}{{\mu^\star {M_t}\sum\nolimits_{l \ne k}^K {\rho _{{\rm{BI}},l}^2N_l^2} }} - \frac{{{\sigma ^2}}}{{{M_t}{M_r}\rho _{{\rm{BI}},k}^2\rho _{{\rm{IU}},k}^2N_k^2}}} \right]}^ + }} \nonumber\\
&< {P_{\max }}
\end{align}
must hold if the dedicated sensing signal is used, i.e., $p_s^ \star  > 0$. On the other hand, if the condition \eqref{sensing_signal_condition} holds, a non-trivial ${p_s^ \star} = {P_{\max }} - \sum\nolimits_{k = 1}^K {{p_{c,k}}\left( \mu ^ \star \right)}$ can be obtained, which satisfies the KKT conditions and is thereby optimal to problem \eqref{C8}.

Then, we have $p_s^ \star  = 0$ provided that condition \eqref{sensing_signal_condition} is not satisfied. In this case, based on \eqref{KKT_condition1} and \eqref{KKT_condition1}, we obtain that
\begin{align}\label{power_com2}
&p_{c,k}^ \star \left( {\mu ,\nu } \right) \nonumber\\
& = {\left[ {\frac{1}{{\mu {M_t}\sum\nolimits_{l \ne k}^K {\rho _{{\rm{BI}},l}^2N_l^2}  + \nu }} \!\!-\!\! \frac{{{\sigma ^2}}}{{{M_t}{M_r}\rho _{{\rm{BI}},k}^2\rho _{{\rm{IU}},k}^2N_k^2}}} \right]^ + }.
\end{align}
In particular, under the case of $\mu^ \star = 0$, we have
\begin{align}\label{power_com3}
p_{c,k}^ \star \left( \nu  \right) = {\left[ {\frac{1}{\nu } - \frac{{{\sigma ^2}}}{{{M_t}{M_r}\rho _{{\rm{BI}},k}^2\rho _{{\rm{IU}},k}^2N_k^2}}} \right]^ + },
\end{align}
where $\nu^\star$ can be obtained by solving an equation $\sum\nolimits_{k = 1}^K {p_{c,k}^ \star \left( \nu  \right)}  = {P_{\max }}$ based on bisection searching. To guarantee that \eqref{power_com3} is optimal, the condition
\begin{align}\label{sensing_signal_condition_suffi}
\sum\limits_{k = 1}^K {{{\left[ {\frac{{{M_t}\rho _{{\rm{BI}},k}^2N_k^2}}{{{\nu ^ \star }}} - \frac{{{\sigma ^2}}}{{{M_r}\rho _{{\rm{IU}},k}^2}}} \right]}^ + } \ge {\Gamma _s}}.
\end{align}
must hold. Otherwise, $\mu^\star > 0$ holds and the optimal ${p_{c.k}}$ is given in \eqref{power_com2}, where $\mu^\star$ and $\nu^\star$ can be obtained by solving a couple of equations, i.e., $\sum\nolimits_{k = 1}^K {p_{c,k}^ \star \left( {\mu ,\lambda } \right)}  = {P_{\max }}$ and $\sum\nolimits_{k = 1}^K {p_{c,k}^ \star \left( {\mu ,\lambda } \right)} {M_t}\rho _{{\rm{BI}},k}^2N_k^2 = {\Gamma _s}$. Thus, we complete the proof.

\bibliographystyle{IEEEtran}

\end{document}